\documentclass[12pt]{article}
\usepackage{amsmath,amssymb,amsthm,mathtools}
\usepackage[margin=1in]{geometry}
\usepackage[numbers,sort&compress]{natbib}
\usepackage[colorlinks=true,linkcolor=blue,citecolor=blue,urlcolor=blue]{hyperref}
\usepackage{enumitem}

\newtheorem{theorem}{Theorem}[section]
\newtheorem{lemma}[theorem]{Lemma}
\newtheorem{proposition}[theorem]{Proposition}
\newtheorem{corollary}[theorem]{Corollary}

\theoremstyle{definition}
\newtheorem{definition}[theorem]{Definition}
\theoremstyle{remark}
\newtheorem{remark}[theorem]{Remark}

\newcommand{\cA}{\mathcal A}
\newcommand{\cB}{\mathcal B}
\newcommand{\cH}{\mathcal H}
\newcommand{\cK}{\mathcal K}
\newcommand{\Chan}{\operatorname{Chan}}
\newcommand{\Ad}{\operatorname{Ad}}
\newcommand{\vN}{\operatorname{vN}}
\newcommand{\id}{\operatorname{id}}
\newcommand{\U}{\mathcal U}
\newcommand{\cb}{\mathrm{cb}}
\newcommand{\NS}{\mathfrak{NS}}
\newcommand{\diam}{\operatorname{diam}}

\title{Essential Duality and Quantitative Channel Composability\\
in Algebraic Quantum Field Theory}
\author{Hassan Nasreddine\thanks{\texttt{hassan.nasreddine@fortisec.net}}}
\date{}

\begin{document}
\maketitle

\begin{abstract}
Let $O\mapsto\cA(O)$ be a local Haag--Kastler net and
$\cA^d(O)=\cA(O')'$ its dual net.  Normal unital completely positive maps on $B(\cH)$
that fix $\cA(O')$ pointwise are exactly the channels with Kraus operators in
$\cA^d(O)$.  Haag duality at $O$ is therefore equality with $\Chan(\cA(O))$, while
essential duality is equivalent to order-independence of these causal-complement channel
classes at spacelike separation.  For self-adjoint generators $a\in M$ and $b\in N$, the
order defect has a finite-parameter spectral formula and mixed second-order coefficient
$\diam\sigma(-i[a,b])$.  Uniform optimization yields
$\Gamma(M,N)=2\Delta_{\rm sa}(M,N)$, and the reversible inner-channel cb balls recover the
same invariant through
\[
\lim_{\varepsilon\downarrow0}
\frac{\Omega^{\rm rev}_\varepsilon(M,N)}{\varepsilon^2}
=\frac14\Gamma(M,N).
\]
The corresponding full-channel coefficient has an exact normalized characterization and
vanishes exactly for commuting algebras.  If a positive self-adjoint energy observable $G$
is fixed, the mean-input-energy-constrained coefficient $\Gamma_{G,E}$ has the same zero
criterion and increases to $\Gamma$ as $E\to\infty$.  No model-independent recovery rate
follows from these hypotheses, even for $G$ with compact resolvent, finite spectral
multiplicities, ground energy zero, and first positive eigenvalue one.  For bosonic
second-quantization nets, the defect at each spacelike pair is either $0$ or $4$: wedge-dual
generalized free fields realize the zero branch despite bounded-region Haag-duality gaps,
whereas nonzero dual symplectic pairing gives extremal Weyl and Clifford witnesses.  In the
translation-covariant case with positive one-particle time generator, failure of essential
duality also admits extremal bounded witness flows that are energy-limited relative to the
second-quantized time-translation energy.
\end{abstract}

\section{Introduction}

In the Haag--Kastler framework \cite{HaagKastler1964,Haag1996}, a bounded region $O$ is
assigned a von Neumann algebra $\cA(O)$, while the causal complement $O'$ carries
$\cA(O')$.  We consider normal channels on $B(\cH)$ that fix $\cA(O')$ pointwise.  The
fixed-commutant characterization of inner channels identifies this class with
\[
\NS_{\cA}(O)
=\Chan\bigl(\cA^d(O)\bigr),
\qquad
\cA^d(O):=\cA(O')'.
\]
The corresponding finite-dimensional commutant statement appears in
\cite[Thm.~5]{Chiribella2018}; the corresponding von Neumann-algebraic extension criterion is
\cite[Prop.~8]{vanLuijkStottmeisterWernerWilming2025}.  This gives two channel formulations
of the duality hierarchy.  Haag duality at $O$ is
$\NS_{\cA}(O)=\Chan(\cA(O))$, whereas essential duality, i.e. locality of the dual net
\cite{BrunettiGuidoLongo1993,HalvorsonMuger2007}, is order-independence of
$\NS_{\cA}(O_1)$ and $\NS_{\cA}(O_2)$ for $O_1\perp O_2$.  A strict inclusion
$\cA(O)\subsetneq\cA^d(O)$ is already detected by complement-fixing unitary flows arbitrarily
close to the identity in cb norm.  This is distinct from the purification characterization of
Haag duality in \cite{vanLuijkStottmeisterWilming2026}.

For von Neumann algebras $M,N\subset B(\cH)$ and self-adjoint contractions $a\in M$,
$b\in N$, set
\[
C_{a,b}(s,t)
=
\bigl\|\Ad_{e^{isa}}\Ad_{e^{itb}}-\Ad_{e^{itb}}\Ad_{e^{isa}}\bigr\|_{\cb}.
\]
The finite-parameter value is determined by the spectrum of the unitary group commutator,
while
\[
\frac{C_{a,b}(s,t)}{|st|}
\longrightarrow
\diam\sigma(-i[a,b])
\qquad ((s,t)\to(0,0),\ st\ne0).
\]
The convergence is uniform on the two self-adjoint unit balls.  This permits optimization
before taking the common-scale limit and gives
\[
\Gamma(M,N)=2\Delta_{\rm sa}(M,N)\in[0,4].
\]
The same coefficient is recovered intrinsically from reversible inner channels near the
identity,
\[
\lim_{\varepsilon\downarrow0}
\frac{\Omega^{\rm rev}_\varepsilon(M,N)}{\varepsilon^2}
=\frac14\Gamma(M,N),
\]
while the full inner-channel classes give a normalized coefficient $\Lambda(M,N)$ with the
same commutativity zero set.  Applied to $M=\cA^d(O_1)$ and $N=\cA^d(O_2)$, these quantities
measure failure of essential duality.

If the representation is equipped with a positive self-adjoint observable $G$, the same
small-parameter analysis can be carried out in the mean-energy-constrained diamond norm
\cite{Shirokov2018,BeckerDattaLamiRouze2021}.  The resulting coefficient satisfies
\[
\Gamma_{G,E}(M,N)=0\iff[M,N]=0,
\qquad
\Gamma_{G,E}(M,N)\uparrow\Gamma(M,N)\quad(E\to\infty).
\]
The energy needed to approach near-extremal spectral edges gives a convenient sufficient
recovery threshold.  A block-diagonal construction shows that the general
von Neumann-algebraic hypotheses determine no recovery rate, even after normalizing the ground
energy and first excitation and imposing compact resolvent with finite spectral multiplicities.

The model analysis uses bosonic second quantization.  The complemented-net correspondence
$K\mapsto\mathcal R(K)$ \cite{Araki1963,EckmannOsterwalder1973,FiglioliniGuido2000} makes the
spacelike defect depend on the symplectic pairing of the corresponding dual localization
subspaces.  A nonzero pairing gives $\Gamma=4$ and an exact Clifford finite-parameter profile,
whereas wedge-dual generalized free fields of Gaier and Yngvason
\cite{GaierYngvason2000} exhibit bounded-region Haag-duality gaps with vanishing spacelike
defect.  For translation-covariant second-quantization nets with positive one-particle time
generator, the nonzero branch can be represented by vectors in every power domain of that
generator, and the associated bounded trigonometric witness flows are energy-limited relative
to its second quantization.  Here $\Chan(\cA(O))$ refers to algebraic Kraus support; localized
probe realizability is a separate operational notion
\cite{FewsterVerch2020,Kitajima2017,MandryschSimmonsNavascues2026}.

Sections~2--3 establish the channel-duality dictionary.  Section~4 develops the quantitative
and energy-constrained coefficients.  Section~5 treats wedge and second-quantization models.
\section{Causal-complement non-signalling channels and the dual algebra}

Let $O\mapsto\cA(O)\subset B(\cH)$ be a Haag--Kastler net of von Neumann algebras indexed by
bounded open double cones in Minkowski space.  We assume isotony and locality, the latter in the
form $[\cA(O_1),\cA(O_2)]=0$ whenever $O_1\subset O_2'$, where $O'$ denotes the causal
complement. We write $O_1\perp O_2$ for $O_1\subset O_2'$.  For an arbitrary open region $R$ we
use the additive extension
\begin{equation}
\cA(R):=\bigvee_{V\subset R}\cA(V),
\label{eq:additive-extension}
\end{equation}
$V$ ranging over double cones contained in $R$, and likewise
$\cA^d(R):=\bigvee_{V\subset R}\cA^d(V)$ whenever an algebra of the dual net is needed for a
general open region.  Set
\begin{equation}
\cA^d(O):=\cA(O')' .
\label{eq:dual-net}
\end{equation}
Locality gives $\cA(O)\subset\cA^d(O)$, and Haag duality at $O$ is equality in this inclusion.

\begin{definition}[Essential duality]
\label{def:essential-duality}
The net $\cA$ satisfies \emph{essential duality} if the dual net $\cA^d$ is local, that is
$[\cA^d(O_1),\cA^d(O_2)]=0$ for every spacelike pair $O_1\perp O_2$ of double cones.
\end{definition}

\begin{remark}
\label{rem:dual-net-haag}
Definition~\ref{def:essential-duality} is the usual one
\cite{BrunettiGuidoLongo1993,HalvorsonMuger2007}, and it is equivalent to Haag duality for the
dual net.  Indeed $\cA(V)\subset\cA^d(V)$ for every double cone $V$, so
$\cA(O')\subset\cA^d(O')$ and hence $\cA^{dd}(O)=\cA^d(O')'\subset\cA(O')'=\cA^d(O)$ for every
$O$.  If $\cA^d$ is local then $\cA^d(O)\subset\cA^d(O')'=\cA^{dd}(O)$ as well, so
$\cA^{dd}=\cA^d$.  Conversely $\cA^{dd}=\cA^d$ gives, for $O_2\subset O_1'$,
$\cA^d(O_2)\subset\cA^d(O_1')\subset\cA^d(O_1)'$.
\end{remark}

\begin{definition}[Causal-complement non-signalling]
\label{def:NS}
A normal unital completely positive map $\Phi:B(\cH)\to B(\cH)$ is
\emph{causal-complement non-signalling at $O$} if
\begin{equation}
(\Phi_*\omega)|_{\cA(O')}=\omega|_{\cA(O')}
\qquad\text{for every normal state }\omega,
\label{eq:NS-state}
\end{equation}
equivalently if
\begin{equation}
\Phi|_{\cA(O')}=\id_{\cA(O')} .
\label{eq:NS-Heis}
\end{equation}
We write $\NS_{\cA}(O)$ for the class of all such channels.
\end{definition}

The equivalence of \eqref{eq:NS-state} and \eqref{eq:NS-Heis} holds because normal states
separate $B(\cH)$.  The assignment is isotone: $O_1\subset O_2$ gives
$\cA(O_2')\subset\cA(O_1')$ and hence $\NS_{\cA}(O_1)\subset\NS_{\cA}(O_2)$.

For an arbitrary Kraus index set, sums are understood as nets over finite subsets.  Thus
$\sum_i k_i^*k_i=1$ means ultraweak convergence of the increasing net of finite partial sums.

\begin{definition}[Inner channels]
For a von Neumann algebra $M\subset B(\cH)$, let $\Chan(M)$ be the class of normal unital
completely positive maps $\Phi$ on $B(\cH)$ admitting a Kraus representation
\begin{equation}
\Phi(x)=\sum_{i\in I} k_i^*xk_i,
\qquad
k_i\in M,
\qquad
\sum_{i\in I}k_i^*k_i=1 .
\label{eq:kraus}
\end{equation}
Such channels are called $M$-inner.
\end{definition}

We use the following fixed-commutant characterization.

\begin{lemma}[Fixed-commutant characterization]
\label{lem:fixed-commutant}
For a von Neumann algebra $M\subset B(\cH)$ and a normal unital completely positive map $\Phi$
on $B(\cH)$,
\begin{equation}
\Phi\in\Chan(M)
\quad\Longleftrightarrow\quad
\Phi|_{M'}=\id_{M'} .
\label{eq:fixed-commutant}
\end{equation}
\end{lemma}

\begin{proof}
If \eqref{eq:kraus} holds with $k_i\in M$ and $y\in M'$, then $y$ commutes with every $k_i$ and
every $k_i^*$, so $\Phi(y)=\sum_i k_i^*yk_i=y\sum_i k_i^*k_i=y$.

Conversely, suppose $\Phi$ fixes $M'$ pointwise.  Take a minimal Stinespring dilation of $\Phi$;
minimality together with normality of $\Phi$ makes the representing homomorphism of $B(\cH)$
normal, and every normal representation of $B(\cH)$ is a multiple of the identity
representation. Hence there are a Hilbert space $\cK$ and an isometry $V:\cH\to\cH\otimes\cK$
with
\[
\Phi(x)=V^*(x\otimes1_{\cK})V .
\]
Every $y\in M'$ lies in the multiplicative domain of $\Phi$: since $M'$ is an algebra,
$\Phi(y^*y)=y^*y=\Phi(y)^*\Phi(y)$, and likewise for $yy^*$.  Writing $P=VV^*$ and substituting
the dilation, the first identity becomes
\[
V^*(y\otimes1_{\cK})^*(1-P)(y\otimes1_{\cK})V=0,
\]
that is $\|(1-P)(y\otimes1_{\cK})V\|=0$.  Hence
\[
(y\otimes1_{\cK})V=P(y\otimes1_{\cK})V=VV^*(y\otimes1_{\cK})V=V\Phi(y)=Vy .
\]
Fix an orthonormal basis $(e_i)_{i\in I}$ of $\cK$ and set $k_i=(1\otimes\langle
e_i,\cdot\rangle)V$, so that $V\xi=\sum_i(k_i\xi)\otimes e_i$.  Comparing basis coefficients in
the displayed intertwining relation gives $yk_i=k_iy$ for every $y\in M'$, so $k_i\in(M')'=M$.
Expanding $V^*(x\otimes1_{\cK})V$ in the basis yields \eqref{eq:kraus}, the finite partial sums
converging ultraweakly (for $x\ge0$ they increase and are bounded by $\|x\|$).
\end{proof}

\begin{lemma}[Algebras from channel classes]
\label{lem:order-embedding}
For von Neumann algebras $M,N\subset B(\cH)$,
\begin{equation}
M\subset N
\quad\Longleftrightarrow\quad
\Chan(M)\subset\Chan(N).
\label{eq:order-embedding}
\end{equation}
\end{lemma}

\begin{proof}
The forward implication is immediate from \eqref{eq:kraus}.  Conversely, let $u\in\U(M)$.  The
unitary channel $\Ad_u(x)=u^*xu$ lies in $\Chan(M)\subset\Chan(N)$ and therefore fixes $N'$
pointwise by Lemma~\ref{lem:fixed-commutant}, so $u\in(N')'=N$.  The unitaries generate $M$.
\end{proof}

\section{Haag duality and Kraus localization}

\begin{proposition}[Dual-algebra reconstruction]
\label{prop:dual-reconstruction}
For every double cone $O$,
\begin{equation}
\NS_{\cA}(O)=\Chan(\cA^d(O)).
\label{eq:C=dualchannels}
\end{equation}
Moreover $\cA^d(O)$ is the smallest von Neumann algebra $M$ for which every channel in
$\NS_{\cA}(O)$ is $M$-inner, and it is already recovered from the reversible members:
\begin{equation}
\cA^d(O)=\vN\{u\in\U(B(\cH)):\Ad_u\in\NS_{\cA}(O)\}.
\label{eq:support-reconstruction}
\end{equation}
\end{proposition}

\begin{proof}
Equation \eqref{eq:C=dualchannels} is Lemma~\ref{lem:fixed-commutant} with
$M=\cA(O')'=\cA^d(O)$. Minimality follows from Lemma~\ref{lem:order-embedding}, since
$\Chan(\cA^d(O))\subset\Chan(M)$ forces $\cA^d(O)\subset M$.  For
\eqref{eq:support-reconstruction}, Lemma~\ref{lem:fixed-commutant} shows that
$\Ad_u\in\NS_{\cA}(O)$ exactly when $u\in\cA(O')'=\cA^d(O)$, and the unitaries of a von Neumann
algebra generate it.
\end{proof}

\paragraph{Unitary-flow characterization.}
For bounded self-adjoint $a\in B(\cH)$,
\begin{equation}
a\in\cA^d(O)
\quad\Longleftrightarrow\quad
\Ad_{e^{ita}}\in\NS_{\cA}(O)\ \text{for every }t\in\mathbb R .
\label{eq:tangent-reconstruction}
\end{equation}
The forward implication follows from \eqref{eq:C=dualchannels}.  Conversely, if every
$\Ad_{e^{ita}}$ fixes $\cA(O')$ pointwise, then $e^{-ita}be^{ita}=b$ for $b\in\cA(O')$ and all
$t$, and differentiation at $t=0$ gives $[a,b]=0$.  So the self-adjoint part of $\cA^d(O)$
consists exactly of the bounded generators of norm-continuous complement-fixing unitary flows at
$O$.

\begin{theorem}[Haag duality and Kraus localization]
\label{thm:haag-kraus}
For a double cone $O$ the following are equivalent.
\begin{enumerate}[label=(\roman*)]
\item $\cA$ satisfies Haag duality at $O$, i.e.\ $\cA(O)=\cA(O')'$.
\item Every normal channel on $B(\cH)$ fixing $\cA(O')$ pointwise is $\cA(O)$-inner.
\item $\NS_{\cA}(O)=\Chan(\cA(O))$.
\end{enumerate}
If Haag duality fails at $O$, then for every bounded self-adjoint
$a\in\cA^d(O)\setminus\cA(O)$ there is $\delta>0$ such that
\begin{equation}
\alpha_t:=\Ad_{e^{ita}}\in\NS_{\cA}(O)\setminus\Chan(\cA(O))
\qquad\text{for all }0<|t|<\delta,
\label{eq:weak-kraus-gap}
\end{equation}
while $\|\alpha_t-\id\|_{\cb}\to0$ as $t\to0$.
\end{theorem}

\begin{proof}
By Proposition~\ref{prop:dual-reconstruction}, (ii) is the inclusion
$\Chan(\cA^d(O))\subset\Chan(\cA(O))$; the reverse inclusion always holds by locality.
Lemma~\ref{lem:order-embedding} therefore gives
$\NS_{\cA}(O)=\Chan(\cA(O))\iff\cA^d(O)=\cA(O)$, which is the equivalence of (i)--(iii).

Suppose Haag duality fails and pick a bounded self-adjoint $a\in\cA^d(O)\setminus\cA(O)$; such
an element exists because a proper inclusion of von Neumann algebras is proper on self-adjoint
parts.  Each $\alpha_t$ is complement-fixing by \eqref{eq:tangent-reconstruction}.  If there
were $t_n\neq0$ with $t_n\to0$ and $\alpha_{t_n}\in\Chan(\cA(O))$, then
Lemma~\ref{lem:fixed-commutant} applied to the unitary channel $\alpha_{t_n}$ would give
$e^{it_na}\in\cA(O)''=\cA(O)$; since
\[
\frac{e^{it_na}-1}{it_n}\longrightarrow a
\quad\text{in operator norm}
\]
and $\cA(O)$ is norm closed and contains $1$, this forces $a\in\cA(O)$, a contradiction.  Hence
\eqref{eq:weak-kraus-gap} holds for some $\delta>0$.  Finally
$\|\Ad_u-\id\|_{\cb}\le2\|u-1\|$ for every unitary $u$, because
$\Ad_u(x)-x=u^*x(u-1)+(u^*-1)x$ and the same identity holds at every matrix level; now use
$e^{ita}\to1$ in norm.
\end{proof}

A strict inclusion $\cA(O)\subsetneq\cA^d(O)$ is therefore already visible in an arbitrarily
small cb neighbourhood of the identity channel.

\section{Essential duality and channel composability}

\begin{lemma}[Channel commutation]
\label{lem:channel-commutation}
For von Neumann algebras $M,N\subset B(\cH)$ the following are equivalent.
\begin{enumerate}[label=(\roman*)]
\item $[M,N]=0$.
\item $\Phi\circ\Psi=\Psi\circ\Phi$ for all $\Phi\in\Chan(M)$ and $\Psi\in\Chan(N)$.
\item $\Ad_u\circ\Ad_v=\Ad_v\circ\Ad_u$ for all $u\in\U(M)$ and $v\in\U(N)$.
\end{enumerate}
\end{lemma}

\begin{proof}
Assume (i) and choose Kraus families $(k_i)_{i\in I}\subset M$ for $\Phi$ and $(\ell_j)_{j\in
J}\subset N$ for $\Psi$.  Then $\Phi\circ\Psi$ has Kraus family $(\ell_jk_i)_{(i,j)\in I\times
J}$ and $\Psi\circ\Phi$ has Kraus family $(k_i\ell_j)$, both nets converging ultraweakly over
finite subsets of $I\times J$.  Since $[M,N]=0$ the two families agree term by term, so the
channels are equal.  Thus (i)$\Rightarrow$(ii), and (ii)$\Rightarrow$(iii) is immediate.

Assume (iii).  Fix a bounded self-adjoint $a\in M$, a unitary $v\in\U(N)$, and put
$u_t=e^{ita}$. With the convention $\Ad_x\circ\Ad_y=\Ad_{yx}$, commutation of $\Ad_{u_t}$ and
$\Ad_v$ reads $\Ad_{vu_t}=\Ad_{u_tv}$.  Two unitaries implement the same automorphism of
$B(\cH)$ exactly when they differ by a scalar phase, since the centre of $B(\cH)$ is trivial.
Hence there is $\lambda(t)\in\mathbb T$ with
\begin{equation}
v^*u_tv=\lambda(t)u_t,
\qquad\text{equivalently}\qquad
\lambda(t)1=v^*u_tvu_t^* .
\label{eq:projective}
\end{equation}
The right-hand identity in \eqref{eq:projective} exhibits $\lambda$ as a norm-continuous
function of $t$ with $\lambda(0)=1$, and
\[
\lambda(s+t)1
=v^*u_su_tvu_t^*u_s^*
=v^*u_sv\,(v^*u_tvu_t^*)\,u_s^*
=\lambda(t)\,v^*u_svu_s^*
=\lambda(s)\lambda(t)1 .
\]
A continuous character of $\mathbb R$ has the form $\lambda(t)=e^{ict}$ with $c\in\mathbb R$, so
$v^*e^{ita}v=e^{it(a+c1)}$ for all $t$.  Differentiating at $t=0$, or using uniqueness of the
bounded generator of a norm-continuous one-parameter unitary group, gives
\[
v^*av=a+c1 .
\]
Unitary conjugation preserves the spectrum while $\sigma(a+c1)=\sigma(a)+c$, and $\sigma(a)$ is
a nonempty compact subset of $\mathbb R$; hence $c=0$ and $v^*av=a$.  As $a$ ranges over the
self-adjoint part of $M$ and $v$ over $\U(N)$, this gives $[M,N]=0$.
\end{proof}

\begin{definition}[Spacelike order-independence]
An assignment $O\mapsto\mathfrak D(O)$ of normal channel classes is \emph{spacelike
order-independent} if $\Phi_1\circ\Phi_2=\Phi_2\circ\Phi_1$ for every spacelike pair
$O_1\perp O_2$ and all $\Phi_i\in\mathfrak D(O_i)$.
\end{definition}

\begin{theorem}[Essential duality and channel order]
\label{thm:essential-composition}
For the assignment $\NS_{\cA}$ of Definition~\ref{def:NS} the following are equivalent.
\begin{enumerate}[label=(\roman*)]
\item $\cA$ satisfies essential duality.
\item $O\mapsto\NS_{\cA}(O)$ is spacelike order-independent.
\item For every spacelike pair $O_1\perp O_2$, the unitary channels in $\NS_{\cA}(O_1)$ commute
under composition with those in $\NS_{\cA}(O_2)$.
\end{enumerate}
\end{theorem}

\begin{proof}
By Proposition~\ref{prop:dual-reconstruction}, $\NS_{\cA}(O)=\Chan(\cA^d(O))$.
Lemma~\ref{lem:channel-commutation} applied to $M=\cA^d(O_1)$ and $N=\cA^d(O_2)$ shows that, for
a given spacelike pair, commutation of the two channel classes, commutation of their unitary
members, and $[\cA^d(O_1),\cA^d(O_2)]=0$ are equivalent.  Quantifying over all spacelike pairs
and using Definition~\ref{def:essential-duality} gives the theorem.
\end{proof}

\begin{remark}[Local extensions]
\label{rem:local-extensions}
The dual net bounds every local extension of $\cA$.  Let $O\mapsto\cB(O)$ be an isotone local
net of von Neumann algebras with $\cA(O)\subset\cB(O)$.  For $x\in\cB(O)$ and a double cone
$V\subset O'$, locality of $\cB$ gives $[x,\cA(V)]\subset[\cB(O),\cB(V)]=0$, so
\begin{equation}
\cB(O)\subset\cA^d(O),
\qquad
\Chan(\cB(O))\subset\NS_{\cA}(O).
\label{eq:universal-operation-net}
\end{equation}
Under essential duality $\cA^d$ is itself local and attains this bound, and is then Haag dual by
Remark~\ref{rem:dual-net-haag}.
\end{remark}

\subsection{Quantitative channel composability}

We quantify the failure of spacelike order-independence along unitary flows in the
causal-complement non-signalling classes.  For bounded self-adjoint $a,b\in B(\cH)$ set
\begin{equation}
C_{a,b}(s,t)
:=\bigl\|\Ad_{e^{isa}}\circ\Ad_{e^{itb}}-\Ad_{e^{itb}}\circ\Ad_{e^{isa}}\bigr\|_{\cb},
\qquad
D_a(x):=\left.\frac{d}{ds}\right|_{s=0}\Ad_{e^{isa}}(x)=-i[a,x].
\label{eq:channel-order-profile}
\end{equation}
We work in the Banach algebra $CB(B(\cH))$ of completely bounded maps, in which
$\Ad_{e^{isa}}=e^{sD_a}$.

\begin{lemma}[Norm of an inner derivation]
\label{lem:derivation-cb-norm}
Let $h=h^*\in B(\cH)$ and $D_h(x)=-i[h,x]$.  Then
\begin{equation}
\|D_h\|=\|D_h\|_{\cb}=\diam\sigma(h)=2\inf_{\lambda\in\mathbb R}\|h-\lambda1\| .
\label{eq:derivation-cb-norm}
\end{equation}
\end{lemma}

\begin{proof}
Write $m=\min\sigma(h)$ and $M=\max\sigma(h)$.  For every $\lambda\in\mathbb R$ and $x\in
B(\cH)$, $\|[h,x]\|=\|[h-\lambda1,x]\|\le2\|h-\lambda1\|\,\|x\|$, and $\lambda=(M+m)/2$ gives
$\|h-\lambda1\|=(M-m)/2$, hence $\|D_h\|\le M-m$.  The same estimate applied to $h\otimes1_n$,
whose spectrum is again $\sigma(h)$, gives $\|D_h\|_{\cb}\le M-m$.

For the lower bound assume $M>m$ and fix $\varepsilon\in(0,(M-m)/2)$.  The spectral projections
of $h$ for $[M-\varepsilon,M]$ and $[m,m+\varepsilon]$ are nonzero, and the two spectral
subspaces are orthogonal and reduce $h$.  Choose unit vectors $\xi,\eta$ in their ranges and set
\[
x=|\xi\rangle\langle\eta|+|\eta\rangle\langle\xi| ,
\]
a self-adjoint contraction with $\|x\|=1$.  Then $x\eta=\xi$ and $\langle\xi,xh\eta\rangle
=\langle\eta,h\eta\rangle$, so
\[
\|D_h(x)\|\ge|\langle\xi,D_h(x)\eta\rangle|
=|\langle\xi,h\xi\rangle-\langle\eta,h\eta\rangle|
\ge M-m-2\varepsilon .
\]
Letting $\varepsilon\downarrow0$ gives $\|D_h\|\ge M-m$, and $\|D_h\|\le\|D_h\|_{\cb}$ closes
the chain.  The witness is a self-adjoint contraction of rank two.  The last equality in
\eqref{eq:derivation-cb-norm} is the midpoint formula for a bounded self-adjoint operator.  The
scalar version of this computation is due to Stampfli \cite{Stampfli1970}.
\end{proof}

Stampfli's theorem holds for an arbitrary bounded operator, and this gives the profile
$C_{a,b}$ in closed form at all parameters.  For a unitary $w\in B(\cH)$ write
\begin{equation}
\nu(w):=\operatorname{dist}\bigl(0,\operatorname{conv}\sigma(w)\bigr),
\label{eq:nu-def}
\end{equation}
and let $\Theta(w)\in[0,\pi]\cup(\pi,2\pi]$ be the length of a shortest closed arc of $\mathbb
T$ containing $\sigma(w)$.  Since $\sigma(w)$ is compact, $\operatorname{conv}\sigma(w)$ is a
compact subset of the plane, so no closure is needed in \eqref{eq:nu-def}.

\begin{proposition}[Finite-scale spectral form of the order defect]
\label{prop:finite-scale-spectral}
Let $a,b\in B(\cH)$ be bounded self-adjoint, put $u=e^{isa}$, $v=e^{itb}$ and
\begin{equation}
w_{a,b}(s,t):=u^*v^*uv=e^{-isa}e^{-itb}e^{isa}e^{itb}.
\label{eq:group-commutator}
\end{equation}
Then, with $w=w_{a,b}(s,t)$,
\begin{equation}
C_{a,b}(s,t)
=\bigl\|\id-\Ad_{w}\bigr\|_{\cb}
=2\inf_{\lambda\in\mathbb C}\|w-\lambda1\|
=2\sqrt{1-\nu(w)^2}
=2\sin\Bigl(\tfrac12\min\{\Theta(w),\pi\}\Bigr).
\label{eq:finite-scale-spectral}
\end{equation}
\end{proposition}

\begin{proof}
With $\Ad_x\circ\Ad_y=\Ad_{yx}$ one has $\Ad_{e^{isa}}\circ\Ad_{e^{itb}}=\Ad_{vu}$ and
$\Ad_{e^{itb}}\circ\Ad_{e^{isa}}=\Ad_{uv}$.  Composing on the right with the surjective complete
isometry $\Ad_{(vu)^*}$ leaves cb norms unchanged and sends $\Ad_{vu}$ to $\id$ and $\Ad_{uv}$
to $\Ad_{(vu)^*(uv)}=\Ad_{w}$, which gives the first equality.

Next, $x-w^*xw=w^*(wx-xw)$, so $\id-\Ad_w=L_{w^*}\circ\delta_w$ with $\delta_w(x)=[w,x]$ and
$L_{w^*}(y)=w^*y$; left multiplication by a unitary is a surjective complete isometry, whence
$\|\id-\Ad_w\|_{\cb}=\|\delta_w\|_{\cb}$.  At matrix level $\delta_w\otimes\id_n$ is the inner
derivation implemented by $w\otimes1_n$, and $\|w\otimes1_n-\lambda1\|=\|w-\lambda1\|$, so
$\|\delta_w\|_{\cb}=\|\delta_w\|$.  Stampfli's theorem \cite{Stampfli1970} for an arbitrary
bounded operator gives $\|\delta_w\|=2\inf_{\lambda\in\mathbb C}\|w-\lambda1\|$.

For the last two equalities, $w$ is normal, so
$\|w-\lambda1\|=\max\{|z-\lambda|:z\in\sigma(w)\}$ and $\inf_\lambda\|w-\lambda1\|$ is the
radius $R$ of the smallest closed disc containing $\sigma(w)$.  Suppose first
$\Theta:=\Theta(w)\le\pi$.  After a rotation,
$\sigma(w)\subset\{e^{i\varphi}:|\varphi|\le\Theta/2\}$ with $e^{\pm i\Theta/2}\in\sigma(w)$.
For $|\varphi|\le\Theta/2$,
\[
\bigl|e^{i\varphi}-\cos(\Theta/2)\bigr|^2
=1-2\cos(\Theta/2)\cos\varphi+\cos^2(\Theta/2)
\le\sin^2(\Theta/2),
\]
so $R\le\sin(\Theta/2)$, while $|e^{i\Theta/2}-e^{-i\Theta/2}|=2\sin(\Theta/2)$ forces equality.
Moreover $\cos\varphi\ge\cos(\Theta/2)$ on that arc gives
$\operatorname{conv}\sigma(w)\subset\{\Re z\ge\cos(\Theta/2)\}$, with the value $\cos(\Theta/2)$
attained at the midpoint of the chord through $e^{\pm i\Theta/2}$; hence $\nu(w)=\cos(\Theta/2)$
and $R=\sqrt{1-\nu(w)^2}$.  If instead $\Theta>\pi$, then $\sigma(w)$ lies in no open half-plane
through $0$, so $0\in\operatorname{conv}\sigma(w)$ and $\nu(w)=0$.  Here $R\le1$; and $R<1$
would give a centre $c$ with $|z-c|\le R$ for all $z\in\sigma(w)$, hence $|c|\ge1-R>0$ and
$\Re(z\bar c)\ge\tfrac12(1+|c|^2-R^2)\ge1-R>0$, putting $\sigma(w)$ in an open half-plane
through $0$ and contradicting $\nu(w)=0$.  So $R=1=\sqrt{1-\nu(w)^2}=\sin(\pi/2)$.
\end{proof}

In finite dimensions the corresponding stabilized distance formula for unitary channels is
standard; see \cite{NechitaPuchalaPawelaZyczkowski2018}.  Expanding
\eqref{eq:group-commutator} gives $w_{a,b}(s,t)=1-st[a,b]+O(|st|(|s|+|t|))$.  For small
$(s,t)$ its principal logarithm gives $w_{a,b}(s,t)=e^{iP_{s,t}}$, with $P_{s,t}=P_{s,t}^*$,
$\|P_{s,t}\|<\pi$, and $P_{s,t}=-st(-i[a,b])+O(|st|(|s|+|t|))$.  Spectral extrema of
self-adjoint operators are norm-Lipschitz, hence
$\Theta(w_{a,b}(s,t))=\diam\sigma(P_{s,t})=|st|\diam\sigma(-i[a,b])+o(|st|)$.  Thus
\eqref{eq:finite-scale-spectral} recovers Theorem~\ref{thm:quantitative-essential}(i); the
uniform statement in part (ii) still requires the exponential-series remainder estimate.

\begin{definition}[Pairwise composability defect]
\label{def:Gamma}
For von Neumann algebras $M,N\subset B(\cH)$ put
\begin{equation}
\Gamma(M,N):=\!\!\sup_{\substack{a=a^*\in M,\ \|a\|\le1\\ b=b^*\in N,\ \|b\|\le1}}\!\!
\|[D_a,D_b]\|_{\cb},
\qquad
\Delta_{\rm sa}(M,N):=\!\!\sup_{\substack{a=a^*\in M,\ \|a\|\le1\\ b=b^*\in N,\ \|b\|\le1}}\!\!
\|[a,b]\| ,
\label{eq:Gamma-Delta-defs}
\end{equation}
and for a spacelike pair $O_1\perp O_2$ set
$\Gamma_{\cA}(O_1,O_2):=\Gamma(\cA^d(O_1),\cA^d(O_2))$.
\end{definition}

\begin{theorem}[Channel-order asymptotics and essential duality]
\label{thm:quantitative-essential}
Let $O_1\perp O_2$ be a spacelike pair.
\begin{enumerate}[label=(\roman*)]
\item For fixed bounded self-adjoint $a\in\cA^d(O_1)$ and $b\in\cA^d(O_2)$,
\begin{equation}
\lim_{\substack{(s,t)\to(0,0)\\ st\neq0}}\frac{C_{a,b}(s,t)}{|st|}
=\|[D_a,D_b]\|_{\cb}
=\diam\sigma(-i[a,b]).
\label{eq:exact-second-order}
\end{equation}
\item At common scale the convergence is uniform on the two self-adjoint unit balls:
\begin{equation}
\Gamma_{\cA}(O_1,O_2)
=\lim_{\varepsilon\downarrow0}\frac{1}{\varepsilon^2}
\sup_{\substack{a=a^*\in\cA^d(O_1),\ \|a\|\le1\\ b=b^*\in\cA^d(O_2),\ \|b\|\le1}}
C_{a,b}(\varepsilon,\varepsilon).
\label{eq:Gamma-common-scale}
\end{equation}
\item $\cA$ satisfies essential duality if and only if $\Gamma_{\cA}(O_1,O_2)=0$ for every
spacelike pair.
\item If essential duality fails, then for some spacelike pair there are self-adjoint contractions
$a\in\cA^d(O_1)$, $b\in\cA^d(O_2)$ with $c:=\diam\sigma(-i[a,b])>0$ such that for every
$\eta\in(0,1)$ there is $\delta>0$ with
\begin{equation}
(1-\eta)c\,|st|\le C_{a,b}(s,t)\le(1+\eta)c\,|st|
\qquad\text{whenever }0<|s|,|t|<\delta .
\label{eq:sharp-order-gap}
\end{equation}
The corresponding flows lie in $\NS_{\cA}(O_1)$ and $\NS_{\cA}(O_2)$ and converge to the
identity in cb norm.
\end{enumerate}
\end{theorem}

\begin{proof}
For self-adjoint $a,b$ the Jacobi identity gives
\begin{equation}
[D_a,D_b](x)=-[[a,b],x]=D_{-i[a,b]}(x),
\label{eq:derivation-bracket}
\end{equation}
and $-i[a,b]$ is self-adjoint, so Lemma~\ref{lem:derivation-cb-norm} yields the second equality
in \eqref{eq:exact-second-order}.

Multiplying the absolutely convergent exponential series in $CB(B(\cH))$ gives
\begin{equation}
\Ad_{e^{isa}}\Ad_{e^{itb}}-\Ad_{e^{itb}}\Ad_{e^{isa}}
=st\,[D_a,D_b]+R_{s,t},
\label{eq:cb-expansion}
\end{equation}
and, with $A=\|D_a\|_{\cb}$ and $B=\|D_b\|_{\cb}$,
\begin{equation}
\|R_{s,t}\|_{\cb}
\le2\,|st|\,AB\bigl(e^{|s|A+|t|B}-1\bigr).
\label{eq:remainder-bound}
\end{equation}
Indeed, after the $st[D_a,D_b]$ term is removed, absolute summation of the remaining mixed
terms gives the displayed bound.  Dividing \eqref{eq:cb-expansion} by $|st|$ and using
\eqref{eq:remainder-bound} proves (i).

If $a,b$ are self-adjoint contractions then $A,B\le2$ by
Lemma~\ref{lem:derivation-cb-norm}, so \eqref{eq:remainder-bound} at $s=t=\varepsilon>0$ gives
\[
\frac{\|R_{\varepsilon,\varepsilon}\|_{\cb}}{\varepsilon^2}\le8\bigl(e^{4\varepsilon}-1\bigr),
\]
a bound independent of $a$ and $b$.  Therefore
$\bigl|\varepsilon^{-2}C_{a,b}(\varepsilon,\varepsilon)-\|[D_a,D_b]\|_{\cb}\bigr|
\le8(e^{4\varepsilon}-1)$ uniformly over the two unit balls, and taking suprema proves (ii).

If essential duality holds, the dual algebras commute at spacelike separation and every
$[D_a,D_b]$ vanishes.  Conversely, suppose $\Gamma_{\cA}(O_1,O_2)=0$.  Then $[D_a,D_b]=0$ for
all self-adjoint $a\in\cA^d(O_1)$, $b\in\cA^d(O_2)$, so by \eqref{eq:derivation-bracket} the
self-adjoint operator $-i[a,b]$ is central in $B(\cH)$ and hence equals $\gamma1$ for some
$\gamma\in\mathbb R$.  If $\gamma\neq0$, then $\frac{d}{dt}\bigl(e^{ita}be^{-ita}\bigr)
=ie^{ita}[a,b]e^{-ita}=-\gamma1$, so $e^{ita}be^{-ita}=b-\gamma t\,1$, which is impossible
because unitary conjugation preserves the nonempty compact spectrum of $b$ while a nonzero real
translation does not.  Hence $[a,b]=0$; since self-adjoint elements span,
$[\cA^d(O_1),\cA^d(O_2)]=0$ for every spacelike pair, which is (iii).

For (iv), failure of essential duality gives by (iii) a spacelike pair with
$\Gamma_{\cA}(O_1,O_2)>0$, hence self-adjoint contractions $a,b$ with $c>0$; the two-sided
estimate is (i), no uniformity in $a,b$ being used.  The flows are complement-fixing by
\eqref{eq:tangent-reconstruction} and satisfy
$\|\Ad_{e^{isa}}-\id\|_{\cb}\le2\|e^{isa}-1\|\to0$.
\end{proof}

\begin{proposition}[Reduction of the optimized defect]
\label{prop:Gamma-Delta}
For von Neumann algebras $M,N\subset B(\cH)$,
\begin{equation}
\Gamma(M,N)
=\!\!\sup_{\substack{a=a^*\in M,\ \|a\|\le1\\ b=b^*\in N,\ \|b\|\le1}}\!\!
\diam\sigma(-i[a,b])
=2\Delta_{\rm sa}(M,N).
\label{eq:Gamma-Delta}
\end{equation}
In particular $\Gamma_{\cA}(O_1,O_2)=2\Delta_{\rm sa}(\cA^d(O_1),\cA^d(O_2))\in[0,4]$ for every
spacelike pair.
\end{proposition}

\begin{proof}
The first equality is \eqref{eq:derivation-bracket} together with
Lemma~\ref{lem:derivation-cb-norm}.  Since $\diam\sigma(h)\le2\|h\|$ for self-adjoint $h$, the
middle term of \eqref{eq:Gamma-Delta} is at most $2\Delta_{\rm sa}(M,N)$.

For the reverse inequality, let $M_{\rm sa,1}$ be the self-adjoint unit ball of $M$.  By the
spectral theorem every $a\in M_{\rm sa,1}$ is a norm limit of operators
$a_0=\sum_{j=1}^m\lambda_jp_j$ with $\lambda_j\in[-1,1]$ and $p_j\in M$ mutually orthogonal
projections summing to $1$.  The vector $(\lambda_1,\dots,\lambda_m)$ lies in the convex hull of
$\{-1,1\}^m$, so $a_0$ is a convex combination of the symmetries $\sum_j\varepsilon_jp_j\in M$,
$\varepsilon_j\in\{-1,1\}$.  Hence $M_{\rm sa,1}$ is the norm-closed convex hull of the
self-adjoint unitaries of $M$, and similarly for $N$.  Bilinearity of the commutator and
convexity of the norm give
\begin{equation}
\Delta_{\rm sa}(M,N)
=\sup\bigl\{\|[s,t]\|:\ s=s^*=s^{-1}\in M,\ t=t^*=t^{-1}\in N\bigr\}.
\label{eq:Delta-symmetries}
\end{equation}

Fix such symmetries $s,t$ and put $u=st$, a unitary.  Then $sus=s(st)s=ts=(st)^*=u^*$, so $u^*$
is unitarily equivalent to $u$ and $\sigma(u)=\overline{\sigma(u)}$.  The continuous function
$f(z)=-i(z-\bar z)=2\operatorname{Im}z$ on $\mathbb T$ is real valued and satisfies $f(\bar
z)=-f(z)$, and $f(u)=-i(u-u^*)=-i[s,t]$.  By the spectral mapping theorem
$\sigma(-i[s,t])=f(\sigma(u))$, a compact subset of $\mathbb R$ symmetric about the origin.  For
a self-adjoint operator with spectrum symmetric about the origin the diameter of the spectrum is
twice the norm, so
\[
\diam\sigma(-i[s,t])=2\|[s,t]\| .
\]
Taking the supremum over symmetries and using \eqref{eq:Delta-symmetries} gives the reverse
inequality.  Finally $\Delta_{\rm sa}(M,N)\le2$, whence the stated range.
\end{proof}

\begin{definition}[Local channel-class order defects]
\label{def:local-channel-defects}
Let $M,N\subset B(\cH)$ be von Neumann algebras.  For $0<\varepsilon<2$ define the reversible
cb-ball defect
\begin{equation}
\Omega^{\mathrm{rev}}_\varepsilon(M,N)
:=
\sup\left\{
\|\alpha\circ\beta-\beta\circ\alpha\|_{\cb}:
\begin{array}{l}
\alpha=\Ad_u,\ u\in\U(M),\quad \beta=\Ad_v,\ v\in\U(N),\\
\|\alpha-\id\|_{\cb}\le\varepsilon,\quad
\|\beta-\id\|_{\cb}\le\varepsilon
\end{array}
\right\}.
\label{eq:Omega-rev-def}
\end{equation}
For $\varepsilon>0$ define the full-channel cb-ball defect
\begin{equation}
\Omega_\varepsilon(M,N)
:=
\sup\left\{
\|\Phi\circ\Psi-\Psi\circ\Phi\|_{\cb}:
\begin{array}{l}
\Phi\in\Chan(M),\quad \Psi\in\Chan(N),\\
\|\Phi-\id\|_{\cb}\le\varepsilon,\quad
\|\Psi-\id\|_{\cb}\le\varepsilon
\end{array}
\right\}.
\label{eq:Omega-full-def}
\end{equation}
\end{definition}

\begin{theorem}[Intrinsic cb-ball realization of the reversible defect]
\label{thm:reversible-cb-tangent}
For von Neumann algebras $M,N\subset B(\cH)$ and $0<\varepsilon<2$, put
$r_\varepsilon:=\arcsin(\varepsilon/2)$.  Then
\begin{equation}
\Omega^{\mathrm{rev}}_\varepsilon(M,N)
=
\sup_{\substack{a=a^*\in M,\ \|a\|\le1\\ b=b^*\in N,\ \|b\|\le1}}
C_{a,b}(r_\varepsilon,r_\varepsilon).
\label{eq:Omega-rev-exact}
\end{equation}
Consequently,
\begin{equation}
\lim_{\varepsilon\downarrow0}
\frac{\Omega^{\mathrm{rev}}_\varepsilon(M,N)}{\varepsilon^2}
=
\frac14\,\Gamma(M,N)
=
\frac12\,\Delta_{\rm sa}(M,N).
\label{eq:reversible-cb-tangent}
\end{equation}
Thus $\Gamma/4$ is the intrinsic quadratic cb-metric tangent of the reversible $M$- and
$N$-inner channel classes at the identity.
\end{theorem}

\begin{proof}
The proof of Proposition~\ref{prop:finite-scale-spectral}, applied to an arbitrary unitary $u$,
gives
\begin{equation}
\|\Ad_u-\id\|_{\cb}
=
2\sin\Bigl(\tfrac12\min\{\Theta(u),\pi\}\Bigr).
\label{eq:unitary-distance-arc}
\end{equation}
Suppose $\|\Ad_u-\id\|_{\cb}\le\varepsilon<2$.  Then
$\Theta(u)\le2r_\varepsilon<\pi$.  After multiplying $u$ by a scalar phase, which does not
change $\Ad_u$, its spectrum lies in
$\{e^{i\theta}:|\theta|\le r_\varepsilon\}$.  Continuous functional calculus on this arc gives
a self-adjoint $h\in M$ with $\|h\|\le r_\varepsilon$ and
$\Ad_u=\Ad_{e^{ih}}$.  Hence $h=r_\varepsilon a$ for a self-adjoint contraction $a\in M$.
Conversely, if $a=a^*\in M$ and $\|a\|\le1$, then
$\Theta(e^{ir_\varepsilon a})\le2r_\varepsilon$, so
\eqref{eq:unitary-distance-arc} gives
$\|\Ad_{e^{ir_\varepsilon a}}-\id\|_{\cb}\le2\sin r_\varepsilon=\varepsilon$.
The same argument applies to $N$, and scalar phases do not affect the order of the implemented
automorphisms.  This proves \eqref{eq:Omega-rev-exact}.

For self-adjoint contractions $a,b$, the uniform estimate in the proof of
Theorem~\ref{thm:quantitative-essential}(ii) gives
\[
\left|
r^{-2}C_{a,b}(r,r)-\|[D_a,D_b]\|_{\cb}
\right|
\le 8(e^{4r}-1),
\]
independently of $a,b$.  Taking suprema and letting $r\downarrow0$ therefore yields
\[
\frac{1}{r^2}
\sup_{\substack{a=a^*\in M,\ \|a\|\le1\\ b=b^*\in N,\ \|b\|\le1}}
C_{a,b}(r,r)
\longrightarrow \Gamma(M,N).
\]
Since $r_\varepsilon/\varepsilon\to1/2$, \eqref{eq:Omega-rev-exact} gives the first equality in
\eqref{eq:reversible-cb-tangent}; the second is Proposition~\ref{prop:Gamma-Delta}.
\end{proof}

\begin{proposition}[Full-channel quadratic cb-ball coefficient]
\label{prop:full-channel-cb-tangent}
For von Neumann algebras $M,N\subset B(\cH)$,
\begin{equation}
0\le \Omega_\varepsilon(M,N)\le 2\varepsilon^2
\qquad(\varepsilon>0).
\label{eq:Omega-full-upper}
\end{equation}
Moreover, $\varepsilon\mapsto\Omega_\varepsilon(M,N)/\varepsilon^2$ is nonincreasing as
$\varepsilon$ increases.  Hence the quadratic cb-ball coefficient
\begin{equation}
\Lambda(M,N):=
\lim_{\varepsilon\downarrow0}
\frac{\Omega_\varepsilon(M,N)}{\varepsilon^2}
\label{eq:Lambda-def}
\end{equation}
exists and satisfies
\begin{equation}
\frac14\,\Gamma(M,N)\le\Lambda(M,N)\le2.
\label{eq:Lambda-bounds}
\end{equation}
It also has the exact normalized characterization
\begin{equation}
\Lambda(M,N)
=
\sup_{\substack{\Phi\in\Chan(M)\setminus\{\id\}\\
                  \Psi\in\Chan(N)\setminus\{\id\}}}
\frac{\|\Phi\circ\Psi-\Psi\circ\Phi\|_{\cb}}
     {\|\Phi-\id\|_{\cb}\,\|\Psi-\id\|_{\cb}},
\label{eq:Lambda-normalized-sup}
\end{equation}
where the supremum of the empty set is understood as $0$.  The following are equivalent:
\[
\begin{aligned}
[M,N]=0
&\quad\Longleftrightarrow\quad
\Omega_{\varepsilon_0}(M,N)=0\ \text{for some }\varepsilon_0>0\\
&\quad\Longleftrightarrow\quad
\Omega_\varepsilon(M,N)=0\ \text{for every }\varepsilon>0\\
&\quad\Longleftrightarrow\quad
\Lambda(M,N)=0.
\end{aligned}
\]
In particular, for a Haag--Kastler net $\cA$, essential duality is equivalent to
\[
\Lambda\bigl(\cA^d(O_1),\cA^d(O_2)\bigr)=0
\]
for every spacelike pair $O_1\perp O_2$; equivalently, the reversible tangent
\eqref{eq:reversible-cb-tangent} vanishes at every such pair.
\end{proposition}

\begin{proof}
Write $A=\Phi-\id$ and $B=\Psi-\id$.  Then
\[
\Phi\circ\Psi-\Psi\circ\Phi=A\circ B-B\circ A,
\]
so every pair admitted in \eqref{eq:Omega-full-def} satisfies
\[
\|\Phi\circ\Psi-\Psi\circ\Phi\|_{\cb}
\le2\|A\|_{\cb}\|B\|_{\cb}
\le2\varepsilon^2,
\]
which proves \eqref{eq:Omega-full-upper}.

Let $0<\delta<\varepsilon$ and set $r=\delta/\varepsilon$.  The classes $\Chan(M)$ and
$\Chan(N)$ are convex and contain $\id$: convex combinations are obtained by adjoining the
correspondingly rescaled Kraus families.  Thus, for any pair $\Phi,\Psi$ admitted at radius
$\varepsilon$,
\[
\Phi_r:=(1-r)\id+r\Phi\in\Chan(M),
\qquad
\Psi_r:=(1-r)\id+r\Psi\in\Chan(N),
\]
and $\|\Phi_r-\id\|_{\cb}\le\delta$, $\|\Psi_r-\id\|_{\cb}\le\delta$.  Direct expansion gives
\[
\Phi_r\circ\Psi_r-\Psi_r\circ\Phi_r
=
r^2(\Phi\circ\Psi-\Psi\circ\Phi).
\]
Taking suprema, with approximation if a supremum is not attained, yields
\[
\frac{\Omega_\delta(M,N)}{\delta^2}
\ge
\frac{\Omega_\varepsilon(M,N)}{\varepsilon^2}.
\]
Together with \eqref{eq:Omega-full-upper}, this proves existence of
\eqref{eq:Lambda-def} and the upper bound in \eqref{eq:Lambda-bounds}.  Since reversible inner
channels are contained in the full inner-channel classes,
$\Omega_\varepsilon(M,N)\ge\Omega^{\mathrm{rev}}_\varepsilon(M,N)$ for $0<\varepsilon<2$.
Theorem~\ref{thm:reversible-cb-tangent} therefore gives the lower bound in
\eqref{eq:Lambda-bounds}.

We next prove \eqref{eq:Lambda-normalized-sup}.  Denote its right-hand side by $S(M,N)$.
For any nonidentity $\Phi\in\Chan(M)$ and $\Psi\in\Chan(N)$, set
$d_\Phi:=\|\Phi-\id\|_{\cb}$ and $d_\Psi:=\|\Psi-\id\|_{\cb}$.  By definition of $S(M,N)$,
\[
\|\Phi\circ\Psi-\Psi\circ\Phi\|_{\cb}
\le S(M,N)d_\Phi d_\Psi.
\]
Hence every pair admitted at radius $\varepsilon$ has order defect at most
$S(M,N)\varepsilon^2$; pairs containing $\id$ have zero defect.  Therefore
$\Omega_\varepsilon(M,N)/\varepsilon^2\le S(M,N)$ and $\Lambda(M,N)\le S(M,N)$.

Conversely, fix nonidentity $\Phi\in\Chan(M)$ and $\Psi\in\Chan(N)$.  For
$0<\varepsilon<\min\{d_\Phi,d_\Psi\}$ put
\[
r_\Phi:=\frac{\varepsilon}{d_\Phi},\qquad
r_\Psi:=\frac{\varepsilon}{d_\Psi},
\]
and define the independent convex dilutions
\[
\widetilde\Phi_\varepsilon:=(1-r_\Phi)\id+r_\Phi\Phi,
\qquad
\widetilde\Psi_\varepsilon:=(1-r_\Psi)\id+r_\Psi\Psi.
\]
They belong to $\Chan(M)$ and $\Chan(N)$, respectively, and
\[
\|\widetilde\Phi_\varepsilon-\id\|_{\cb}
=
\|\widetilde\Psi_\varepsilon-\id\|_{\cb}
=\varepsilon.
\]
A direct expansion gives
\[
\widetilde\Phi_\varepsilon\circ\widetilde\Psi_\varepsilon
-
\widetilde\Psi_\varepsilon\circ\widetilde\Phi_\varepsilon
=
\frac{\varepsilon^2}{d_\Phi d_\Psi}
(\Phi\circ\Psi-\Psi\circ\Phi).
\]
Thus, for all sufficiently small $\varepsilon$,
\[
\frac{\Omega_\varepsilon(M,N)}{\varepsilon^2}
\ge
\frac{\|\Phi\circ\Psi-\Psi\circ\Phi\|_{\cb}}
     {d_\Phi d_\Psi}.
\]
Letting $\varepsilon\downarrow0$ and then taking the supremum over $\Phi,\Psi$ proves
$\Lambda(M,N)\ge S(M,N)$, and hence \eqref{eq:Lambda-normalized-sup}.

If $[M,N]=0$, Lemma~\ref{lem:channel-commutation} shows that every channel in $\Chan(M)$
commutes with every channel in $\Chan(N)$, so $\Omega_\varepsilon=0$ for all $\varepsilon$.
Conversely, if $\Omega_{\varepsilon_0}(M,N)=0$ for some $\varepsilon_0>0$, then
$\Omega_\delta(M,N)=0$ for every $0<\delta<\varepsilon_0$, hence $\Lambda(M,N)=0$.
Finally, $\Lambda(M,N)=0$ and \eqref{eq:Lambda-bounds} imply $\Gamma(M,N)=0$, so
$[M,N]=0$ by Proposition~\ref{prop:Gamma-Delta}.  The AQFT statement follows from
Theorem~\ref{thm:essential-composition}.
\end{proof}

\begin{corollary}[Expectation-value realization]
\label{cor:expectation-realization}
Let $a,b\in B(\cH)$ be bounded self-adjoint operators and put
$h=-i[a,b]$.  For a normal state $\omega$ and a self-adjoint contraction $x$, define
\[
F_{\omega,x}^{a,b}(s,t)
:=\omega\Bigl(\bigl(\Ad_{e^{isa}}\Ad_{e^{itb}}-\Ad_{e^{itb}}\Ad_{e^{isa}}\bigr)(x)\Bigr).
\]
Then
\begin{equation}
\diam\sigma(h)
=\sup_{\substack{x=x^*,\ \|x\|\le1\\ \omega\ \mathrm{normal\ state}}}
\left|\left.\frac{\partial^2}{\partial s\,\partial t}\right|_{(0,0)}
F_{\omega,x}^{a,b}(s,t)\right| .
\label{eq:expectation-realization}
\end{equation}
The supremum is unchanged if normal states are replaced by vector states.  Suppose
$c:=\diam\sigma(h)>0$ and fix $\eta\in(0,1)$.  There are a self-adjoint contraction $x$, a
normal vector state $\omega$, and $\delta>0$ such that
\begin{equation}
\bigl|F_{\omega,x}^{a,b}(s,t)\bigr|
\ge (1-\eta)c\,|st|
\qquad (0<|s|,|t|<\delta).
\label{eq:expectation-near-witness}
\end{equation}
\end{corollary}

\begin{proof}
By \eqref{eq:cb-expansion} and \eqref{eq:derivation-bracket},
\[
\left.\frac{\partial^2}{\partial s\,\partial t}\right|_{(0,0)}
F_{\omega,x}^{a,b}(s,t)
=\omega(D_h(x))=-\omega([[a,b],x]).
\]
Since $h$ and $x$ are self-adjoint, so is $D_h(x)$, and therefore
$\sup_\omega|\omega(D_h(x))|=\|D_h(x)\|$, with the same value over normal states and over vector
states.  The right-hand side of \eqref{eq:expectation-realization} is thus the norm of $D_h$ on
the self-adjoint unit ball.  It is at most $\|D_h\|=\diam\sigma(h)$, and the witness constructed
in the proof of Lemma~\ref{lem:derivation-cb-norm} is a self-adjoint contraction, so the two
agree.

For \eqref{eq:expectation-near-witness}, choose $x$ and a vector state $\omega$ with
$|\omega(D_h(x))|>(1-\eta/2)c$.  The expansion \eqref{eq:cb-expansion} gives
\[
F_{\omega,x}^{a,b}(s,t)=st\,\omega(D_h(x))+\omega(R_{s,t}(x)).
\]
For the fixed pair $a,b$, \eqref{eq:remainder-bound} implies
$\|R_{s,t}\|_{\cb}/|st|\to0$ as $(s,t)\to(0,0)$; shrinking $\delta$ so that
$\|R_{s,t}\|_{\cb}\le(\eta/2)c\,|st|$ proves \eqref{eq:expectation-near-witness}.
\end{proof}

For a spacelike pair $O_1\perp O_2$, Corollary~\ref{cor:expectation-realization} shows that
essential duality is equivalent to the vanishing of every mixed expectation-value order response
associated with bounded self-adjoint generators in $\cA^d(O_1)$ and $\cA^d(O_2)$.

\subsection{Energy-constrained dual-net order defect}
\label{subsec:energy-constrained}

Assume in this subsection that the chosen representation is equipped with a positive self-adjoint
energy observable $G$.  For $E>\inf\sigma(G)$ put
\[
\mathfrak S_E^G(\cH)
:=\{\rho\ge0:\ \rho\text{ trace class},\ \operatorname{Tr}\rho=1,
\ \operatorname{Tr}(G\rho)\le E\}.
\]
For a Hermiticity-preserving map $T$ on trace class, let
\begin{equation}
\|T\|_{\diamond,G,E}
:=\sup\left\{\|(T\otimes\id)(\varrho)\|_1:
\varrho\text{ a state on }\cH\otimes\cH_R,
\ \operatorname{Tr}(G\varrho_{\cH})\le E\right\},
\label{eq:EC-diamond-def}
\end{equation}
with unrestricted reference system.  This is the standard mean-energy-constrained diamond norm
\cite{Shirokov2018,BeckerDattaLamiRouze2021}.  For bounded self-adjoint $a,b$, define
\[
\widehat\alpha_s(\rho)=e^{isa}\rho e^{-isa},\qquad
\widehat\beta_t(\rho)=e^{itb}\rho e^{-itb},
\]
and
\begin{equation}
C^{G,E}_{a,b}(s,t)
:=\|\widehat\beta_t\widehat\alpha_s-\widehat\alpha_s\widehat\beta_t\|_{\diamond,G,E}.
\label{eq:EC-order-defect}
\end{equation}

\begin{theorem}[Energy-constrained dual-net order defect]
\label{thm:EC-dual-net-defect}
Let $M,N\subset B(\cH)$ be von Neumann algebras, let $E>\inf\sigma(G)$, and for bounded
self-adjoint $h$ set
\[
\gamma_{G,E}(h)
:=2\sup_{\rho\in\mathfrak S_E^G(\cH)}\sqrt{\operatorname{Var}_{\rho}(h)},
\qquad
\operatorname{Var}_{\rho}(h)=\operatorname{Tr}(\rho h^2)-\operatorname{Tr}(\rho h)^2.
\]
Then:
\begin{enumerate}[label=(\roman*)]
\item For fixed bounded self-adjoint $a,b$, with $h=-i[a,b]$,
\begin{equation}
\lim_{\substack{(s,t)\to(0,0)\\st\neq0}}
\frac{C^{G,E}_{a,b}(s,t)}{|st|}=\gamma_{G,E}(h).
\label{eq:EC-fixed-generator-limit}
\end{equation}
\item Defining
\begin{equation}
\Gamma_{G,E}(M,N)
:=\sup_{\substack{a=a^*\in M,\ \|a\|\le1\\ b=b^*\in N,\ \|b\|\le1}}
\gamma_{G,E}(-i[a,b]),
\label{eq:Gamma-GE-def}
\end{equation}
one has the uniform common-scale formula
\begin{equation}
\Gamma_{G,E}(M,N)
=\lim_{\varepsilon\downarrow0}\frac1{\varepsilon^2}
\sup_{\substack{a=a^*\in M,\ \|a\|\le1\\ b=b^*\in N,\ \|b\|\le1}}
C^{G,E}_{a,b}(\varepsilon,\varepsilon).
\label{eq:Gamma-GE-common-scale}
\end{equation}
\item For every $E>\inf\sigma(G)$,
\begin{equation}
0\le\Gamma_{G,E}(M,N)\le\Gamma(M,N)\le4,
\qquad
\Gamma_{G,E}(M,N)=0\iff[M,N]=0,
\label{eq:Gamma-GE-zero}
\end{equation}
and $E\mapsto\Gamma_{G,E}(M,N)$ is nondecreasing.
\item The constrained coefficients recover the algebraic one:
\begin{equation}
\Gamma_{G,E}(M,N)\uparrow\Gamma(M,N)\qquad(E\to\infty).
\label{eq:Gamma-GE-limit}
\end{equation}
\end{enumerate}
Consequently, if
$\Gamma_{\cA;G,E}(O_1,O_2):=\Gamma_{G,E}(\cA^d(O_1),\cA^d(O_2))$, then for every fixed
$E>\inf\sigma(G)$ essential duality is equivalent to
$\Gamma_{\cA;G,E}(O_1,O_2)=0$ at all spacelike pairs.
\end{theorem}

\begin{proof}
We first identify the constrained norm of the infinitesimal generator.  For
$\mathcal L_h(\rho)=i[h,\rho]$ and an admissible state $\varrho$, purify $\varrho$ to a unit
vector $\Psi$ without changing its system marginal.  Trace-norm contractivity under the partial
trace and the rank-two form of the commutator give
\[
\|(\mathcal L_h\otimes\id)(\varrho)\|_1
\le2\sqrt{\operatorname{Var}_{\varrho_{\cH}}(h)}.
\]
Conversely, purifying any $\rho\in\mathfrak S_E^G(\cH)$ gives equality in the same rank-two
calculation.  Hence
\begin{equation}
\|\mathcal L_h\|_{\diamond,G,E}=\gamma_{G,E}(h).
\label{eq:EC-commutator-variance}
\end{equation}

On trace class let $L_a(\rho)=i[a,\rho]$ and $L_b(\rho)=i[b,\rho]$.  Then
$\widehat\alpha_s=e^{sL_a}$, $\widehat\beta_t=e^{tL_b}$ and $[L_b,L_a]=\mathcal L_{-i[a,b]}$.
The same absolutely convergent exponential-series estimate as in
\eqref{eq:remainder-bound}, now in the Banach algebra of bounded maps on trace class, gives
\[
\widehat\beta_t\widehat\alpha_s-\widehat\alpha_s\widehat\beta_t
=st\,\mathcal L_{-i[a,b]}+\widehat R_{s,t},
\qquad
\|\widehat R_{s,t}\|_\diamond
\le2|st|AB\bigl(e^{|s|A+|t|B}-1\bigr),
\]
where $A=\|L_a\|_\diamond$ and $B=\|L_b\|_\diamond$.  Since
$\|\cdot\|_{\diamond,G,E}\le\|\cdot\|_\diamond$, equation
\eqref{eq:EC-commutator-variance} proves (i).  If $a,b$ are self-adjoint contractions, then
$A,B\le2$, so at $s=t=\varepsilon$
\[
\left|\varepsilon^{-2}C^{G,E}_{a,b}(\varepsilon,\varepsilon)
-\gamma_{G,E}(-i[a,b])\right|
\le8(e^{4\varepsilon}-1),
\]
uniformly in $a,b$ (and in $E$), proving (ii).

For any bounded self-adjoint $h$,
\[
\operatorname{Var}_{\rho}(h)\le\frac14\diam\sigma(h)^2,
\qquad
\gamma_{G,E}(h)\le\diam\sigma(h).
\]
This gives the inequalities in (iii), while monotonicity in $E$ is immediate.  Moreover
\begin{equation}
\gamma_{G,E}(h)=0\iff h\text{ is scalar}.
\label{eq:gamma-GE-scalar}
\end{equation}
Indeed, if $h$ is non-scalar, choose a finite-energy unit vector $\psi$ with positive variance
(using density of $\operatorname{dom}G^{1/2}$), and choose a state $\rho_0$ with
$\operatorname{Tr}(G\rho_0)<E$.  For sufficiently small $p>0$ the mixture
$\rho_p=(1-p)\rho_0+p|\psi\rangle\langle\psi|$ is admissible.  Since variance is concave in
the state,
\[
\operatorname{Var}_{\rho_p}(h)\ge(1-p)\operatorname{Var}_{\rho_0}(h)
+p\operatorname{Var}_{\psi}(h)>0.
\]  Thus
$\Gamma_{G,E}(M,N)=0$ implies that $-i[a,b]$ is scalar for every self-adjoint $a\in M$, $b\in N$.
As in Theorem~\ref{thm:quantitative-essential}(iii), a nonzero scalar commutator of bounded
self-adjoint operators would translate the compact spectrum of $b$ under conjugation by
$e^{ita}$, which is impossible.  Hence $[M,N]=0$; the converse is immediate.

Finally, spectral cutoffs $\mathbf1_{[0,R]}(G)$ converge strongly to $1$, so finite-energy unit
vectors are dense.  Since $h$ and $h^2$ are bounded, vector-state variances are continuous under
norm approximation, while the unrestricted maximal variance equals
$\frac14\diam\sigma(h)^2$.  Therefore
$\gamma_{G,E}(h)\uparrow\diam\sigma(h)$ for every bounded self-adjoint $h$.  Choosing
self-adjoint contractions $a,b$ whose spectral coefficient is arbitrarily close to
$\Gamma(M,N)$ proves (iv).  The AQFT statement follows from
Theorem~\ref{thm:essential-composition}.
\end{proof}

\paragraph{Relation to Haag duality.}
If Haag duality holds at both members of a spacelike pair, locality gives
$\Gamma_{\cA}(O_1,O_2)=0$.  Hence a positive defect requires a Haag-duality gap at least
one of the two regions; Proposition~\ref{prop:GY-zero-defect} shows that the converse need not
hold.

\subsection{Energy accessibility and recovery rates}
\label{subsec:energy-accessibility}

For a bounded self-adjoint $h$, write
\[
\lambda_+(h)=\max\sigma(h),\qquad
\lambda_-(h)=\min\sigma(h),
\]
and, for a unit vector $\psi\in\operatorname{dom}G^{1/2}$,
$\mathcal E_G(\psi)=\|G^{1/2}\psi\|^2$.  For $\delta>0$ set
\begin{equation}
\mathfrak e_{G,h}(\delta)
:=\inf\left\{
\frac{\mathcal E_G(\psi_+)+\mathcal E_G(\psi_-)}{2}:
\begin{array}{l}
\psi_\pm\in\operatorname{dom}G^{1/2},\quad \|\psi_\pm\|=1,\\
\langle\psi_+,h\psi_+\rangle\ge\lambda_+(h)-\delta,\\
\langle\psi_-,h\psi_-\rangle\le\lambda_-(h)+\delta
\end{array}
\right\}.
\label{eq:edge-accessibility}
\end{equation}
For $\varepsilon>0$ define
\begin{equation}
\mathfrak E^{\rm edge}_{G;M,N}(\varepsilon)
:=\inf\left\{
\mathfrak e_{G,-i[a,b]}(\varepsilon/3):
\begin{array}{l}
a=a^*\in M,\ \|a\|\le1,\qquad b=b^*\in N,\ \|b\|\le1,\\
\diam\sigma(-i[a,b])>\Gamma(M,N)-\varepsilon/3
\end{array}
\right\}.
\label{eq:optimized-edge-profile}
\end{equation}

\begin{proposition}[Spectral-edge recovery]
\label{prop:spectral-edge-accessibility}
For every bounded self-adjoint $h$ and every $\delta>0$,
$\mathfrak e_{G,h}(\delta)<\infty$, and
\begin{equation}
E>\mathfrak e_{G,h}(\delta)
\quad\Longrightarrow\quad
\gamma_{G,E}(h)\ge\diam\sigma(h)-2\delta.
\label{eq:edge-accessibility-bound}
\end{equation}
Consequently, for every $\varepsilon>0$,
\begin{equation}
E>\mathfrak E^{\rm edge}_{G;M,N}(\varepsilon)
\quad\Longrightarrow\quad
\Gamma_{G,E}(M,N)>\Gamma(M,N)-\varepsilon.
\label{eq:finite-energy-recovery}
\end{equation}
\end{proposition}

\begin{proof}
The spectral subspaces of $h$ associated with
$(\lambda_+(h)-\delta/2,\lambda_+(h)]$ and
$[\lambda_-(h),\lambda_-(h)+\delta/2)$ are nonzero.  Unit vectors in these subspaces can be
approximated in norm by unit vectors in $\operatorname{dom}G^{1/2}$; boundedness of $h$ then
gives admissible vectors in \eqref{eq:edge-accessibility}.  Thus
$\mathfrak e_{G,h}(\delta)<\infty$.

If $E>\mathfrak e_{G,h}(\delta)$, choose admissible $\psi_\pm$ whose average energy is below
$E$ and put
\[
\rho=\tfrac12|\psi_+\rangle\langle\psi_+|
     +\tfrac12|\psi_-\rangle\langle\psi_-|.
\]
For $\mu_\pm=\langle\psi_\pm,h\psi_\pm\rangle$,
\[
\operatorname{Var}_\rho(h)
=
\tfrac12\operatorname{Var}_{\psi_+}(h)
+\tfrac12\operatorname{Var}_{\psi_-}(h)
+\tfrac14(\mu_+-\mu_-)^2.
\]
This gives \eqref{eq:edge-accessibility-bound}; when
$2\delta\ge\diam\sigma(h)$ the bound is immediate from $\gamma_{G,E}(h)\ge0$.
The defining set in \eqref{eq:optimized-edge-profile} is nonempty by
\eqref{eq:Gamma-Delta}.  If $E$ is larger than its infimum, choose $a,b$ with
\[
\diam\sigma(-i[a,b])>\Gamma(M,N)-\varepsilon/3,
\qquad
\mathfrak e_{G,-i[a,b]}(\varepsilon/3)<E.
\]
Applying \eqref{eq:edge-accessibility-bound} with $\delta=\varepsilon/3$ and then taking the
supremum in \eqref{eq:Gamma-GE-def} proves \eqref{eq:finite-energy-recovery}.
\end{proof}

The profile in \eqref{eq:optimized-edge-profile} is a sufficient energy threshold; an upper
bound on it in a specific model gives a corresponding convergence estimate.  No such rate
follows from the hypotheses of Theorem~\ref{thm:EC-dual-net-defect} alone.

\begin{proposition}[No model-independent recovery rate]
\label{prop:no-uniform-energy-rate}
For every sequence $r_k\downarrow0$ with $0<r_k<1$, there are von Neumann algebras
$M,N\subset B(\cH)$, a positive self-adjoint $G$ with compact resolvent, finite spectral
multiplicities, $\inf\sigma(G)=0$, and first positive eigenvalue equal to $1$, together with
energies $E_k\to\infty$, such that
\begin{equation}
\Gamma(M,N)=4,
\qquad
4-\Gamma_{G,E_k}(M,N)\ge r_k
\quad(k\ge3).
\label{eq:arbitrarily-slow-recovery}
\end{equation}
Consequently there is no function $R:[0,\infty)\to[0,\infty)$ with $R(E)\to0$, independent of $G,M,N$, for which
\begin{equation}
\Gamma(M,N)-\Gamma_{G,E}(M,N)\le R(E)
\label{eq:no-uniform-rate-form}
\end{equation}
holds for all such triples and all sufficiently large $E$.
\end{proposition}

\begin{proof}
Let
\[
\cH=\bigoplus_{n\ge1}\mathbb C^2,
\qquad
G=\bigoplus_{n\ge1}E_n1_2,
\]
where $E_1=0$, $E_2=1$, $E_3>1$, and, for $k\ge3$,
\begin{equation}
E_{k+1}\ge\frac{16E_k}{r_k^2}.
\label{eq:energy-ladder-choice}
\end{equation}
Put $s_1=s_2=0$ and $s_k=1-r_k/2$ for $k\ge3$.  Choose
$\theta_k\in[0,\pi/4]$ with $\sin(2\theta_k)=s_k$, and let
\[
Z=\begin{pmatrix}1&0\\0&-1\end{pmatrix},\qquad
X=\begin{pmatrix}0&1\\1&0\end{pmatrix},\qquad
B_k=\cos(2\theta_k)Z+\sin(2\theta_k)X.
\]
Define the block-diagonal von Neumann algebras
\[
\begin{aligned}
M&=\left\{\bigoplus_n a_n:\ a_n\in\vN(Z),\ \sup_n\|a_n\|<\infty\right\},\\
N&=\left\{\bigoplus_n b_n:\ b_n\in\vN(B_n),\ \sup_n\|b_n\|<\infty\right\}.
\end{aligned}
\]
Since $E_n\to\infty$ and every eigenvalue has multiplicity two, $G$ has compact resolvent and
finite spectral multiplicities, with the stated normalization.

For a self-adjoint contraction $a_n\in\vN(Z)$ and a self-adjoint contraction
$b_n\in\vN(B_n)$, write
$a_n=\alpha_n1+\beta_nZ$ and $b_n=\gamma_n1+\delta_nB_n$.  Then
$|\beta_n|,|\delta_n|\le1$, and, with
$Y=\begin{psmallmatrix}0&-i\\ i&0\end{psmallmatrix}$,
\[
-i[a_n,b_n]=2\beta_n\delta_n s_nY,
\qquad
(-i[a_n,b_n])^2\le4s_n^2 1_2.
\]
If $\rho$ is a state with $\operatorname{Tr}(G\rho)\le E$ and $P_n$ is the projection onto
the $n$th block, put $p_n=\operatorname{Tr}(P_n\rho)$.  Then
$\sum_n p_n=1$, $\sum_nE_np_n\le E$, and
\[
\operatorname{Var}_\rho(-i[a,b])
\le\operatorname{Tr}\rho(-i[a,b])^2
\le4\sum_n p_ns_n^2.
\]
Conversely, take $a=\bigoplus_nZ$ and $b=\bigoplus_nB_n$.  For every probability vector
$(p_n)$ with $\sum_nE_np_n\le E$, the state
$\rho_p=\bigoplus_n p_n1_2/2$ has mean energy at most $E$, zero expectation of
$-i[a,b]$, and variance $4\sum_np_ns_n^2$.  Hence
\begin{equation}
\Gamma_{G,E}(M,N)
=4\left(
\sup_{\substack{p_n\ge0,\ \sum_np_n=1\\ \sum_nE_np_n\le E}}
\sum_n p_ns_n^2
\right)^{1/2}.
\label{eq:block-energy-exact}
\end{equation}
Since $s_k\uparrow1$, this formula also gives $\Gamma(M,N)=4$.

At $E=E_k$, let $q=\sum_{n>k}p_n$.  The energy constraint and monotonicity of the
$E_n$ give $q\le E_k/E_{k+1}$, while $s_n\le s_k$ for $n\le k$.  Therefore
\[
\sum_n p_ns_n^2
\le s_k^2+\frac{E_k}{E_{k+1}},
\]
and \eqref{eq:block-energy-exact} together with \eqref{eq:energy-ladder-choice} yields
\[
\Gamma_{G,E_k}(M,N)
\le4s_k+4\sqrt{\frac{E_k}{E_{k+1}}}
\le4-2r_k+r_k=4-r_k.
\]
This proves \eqref{eq:arbitrarily-slow-recovery}.

For the final assertion, given a candidate $R(E)\to0$, take for instance $r_k=2^{-k}$.
The eigenvalues can be chosen recursively so that, in addition to
\eqref{eq:energy-ladder-choice}, $R(E_k)<r_k$ for every $k\ge3$.  Then
\eqref{eq:arbitrarily-slow-recovery} contradicts \eqref{eq:no-uniform-rate-form} along the
sequence $E_k$.
\end{proof}

The construction lies in the general von Neumann-algebraic setting of
Theorem~\ref{thm:EC-dual-net-defect}.  It shows that a quantitative recovery rate requires
additional structure relating the energy observable to the two algebras.

\section{Wedge realization and model examples}
\label{sec:wedge-envelope}

The examples below realize complementary regimes within the second-quantization framework.
Wedge-dual generalized free fields separate Haag duality from essential duality, while the Weyl
criterion identifies when the nonzero quantitative branch occurs.  Wedge duality is used only in
this section; the results of Sections~2--4 do not assume it.  We use the standard wedge-separation
property of Minkowski space: for every spacelike pair of double cones $O_1\perp O_2$ there is a
wedge $W$ such that
\begin{equation}
O_1\subset W\subset O_2',
\label{eq:wedge-separation}
\end{equation}
equivalently $O_1\subset W$ and $O_2\subset W'$ \cite{ThomasWichmann1997}.

For a double cone $O$ define its wedge envelope by
\begin{equation}
\cA^{\wedge}(O):=\bigcap_{W\supset O}\cA(W),
\label{eq:wedge-envelope}
\end{equation}
where the intersection runs over wedges containing $O$.

\begin{lemma}[Wedge realization of the dual algebra]
\label{lem:wedge-envelope-dual}
Assume wedge duality,
\begin{equation}
\cA(W')=\cA(W)'
\qquad\text{for every wedge }W.
\label{eq:wedge-duality}
\end{equation}
Then every double cone $O$ satisfies
\begin{equation}
\cA^{\wedge}(O)=\cA(O')'=\cA^d(O),
\label{eq:wedge-dual-identity}
\end{equation}
and consequently $\NS_{\cA}(O)=\Chan(\cA^{\wedge}(O))$.
\end{lemma}

\begin{proof}
We first show
\begin{equation}
\cA(O')=\bigvee_{W\supset O}\cA(W').
\label{eq:complement-from-opposite-wedges}
\end{equation}
Let $V\subset O'$ be a double cone.  Applying the wedge-separation property
\eqref{eq:wedge-separation} to $O$ and $V$ gives a wedge $W$ with $O\subset W\subset V'$, hence
$V\subset W'$ and $\cA(V)\subset\cA(W')$.  Since $\cA(O')$ is the join of the algebras $\cA(V)$
over all double cones $V\subset O'$, this proves the inclusion ``$\subset$'' in
\eqref{eq:complement-from-opposite-wedges}.  Conversely, $O\subset W$ implies $W'\subset O'$, so
$\cA(W')\subset\cA(O')$ by isotony of the open-region extension.  Taking commutants and using
wedge duality gives
\[
\cA(O')'=\bigcap_{W\supset O}\cA(W')'
=\bigcap_{W\supset O}\cA(W)=\cA^{\wedge}(O),
\]
which is \eqref{eq:wedge-dual-identity}.  The channel identity is
Proposition~\ref{prop:dual-reconstruction}.
\end{proof}

\begin{lemma}[Wedge duality gives zero defect]
\label{lem:wedge-zero-defect}
Assume wedge duality.  Then for every spacelike pair $O_1\perp O_2$,
\[
[\cA^d(O_1),\cA^d(O_2)]=0,
\qquad
\Gamma_{\cA}(O_1,O_2)=0.
\]
In particular the dual net is local.
\end{lemma}

\begin{proof}
Choose a wedge $W$ as in \eqref{eq:wedge-separation}.  Since $O_1\subset W$ and
$O_2\subset W'$, causal complementation reverses inclusion and gives
$W'\subset O_1'$ and $W\subset O_2'$.  Therefore
\[
\cA^d(O_1)=\cA(O_1')'\subset\cA(W')'=\cA(W),
\qquad
\cA^d(O_2)=\cA(O_2')'\subset\cA(W)'=\cA(W').
\]
The two dual algebras commute.  The vanishing of $\Gamma_{\cA}$ follows from
Theorem~\ref{thm:quantitative-essential}.
\end{proof}

\begin{proposition}[Zero defect with a Haag-duality gap]
\label{prop:GY-zero-defect}
Consider the generalized free field models of Gaier and Yngvason
\cite[Examples 1 and 2 in Sec.~5]{GaierYngvason2000}: a massless field in $d=3$ with
$M(p)=(a\cdot p)^{2n}$ for a spacelike or lightlike $a\neq0$, and the corresponding $d=2$ models
with $m>0$.  Let $\cA=\mathcal M_{\min}$ be the minimal net generated by the field.  Then
\begin{enumerate}[label=\textup{(\roman*)}]
\item $\cA$ satisfies essential duality and
$\Gamma_{\cA}(O_1,O_2)=0$ for every spacelike pair $O_1\perp O_2$;
\item Haag duality fails at every bounded region $O$;
\item at every bounded region $O$ there are complement-fixing unitary channels arbitrarily close
to the identity in cb norm which are not $\cA(O)$-inner.
\end{enumerate}
Thus $\NS_{\cA}(O)$ is strictly larger than $\Chan(\cA(O))$ at every bounded region while the
spacelike composability defect vanishes identically.
\end{proposition}

\begin{proof}
In these examples the rational function $p_+\mapsto M(p_+,p_+^{-1}(\hat p^2+m^2),\hat p)$ has
only real zeros, and the same holds after every Lorentz transformation because $M_\Lambda$ is
again of the stated form; by \cite[Prop.~4.1]{GaierYngvason2000} wedge duality holds for every
wedge.  Gaier and Yngvason also record that
\[
\mathcal M_{\min}(O)\neq\mathcal M_{\max}(O),
\qquad
\mathcal M_{\max}(O)=\mathcal M_{\min}(O')'
=\bigcap_{W\supset O}\mathcal M_{\min}(W)
\]
for bounded $O$.  On every wedge, however, $\mathcal M_{\min}(W)=\mathcal M_{\max}(W)$
\cite[Remark following Eq.~(53)]{GaierYngvason2000}.  The identification of
$\mathcal M_{\min}(O')'$ with the intersection over wedges is
Lemma~\ref{lem:wedge-envelope-dual} applied to $\mathcal M_{\min}$.  Since
$\mathcal M_{\max}(O)=\mathcal M_{\min}(O')'$ is the dual-net algebra $\cA^d(O)$ of the
minimal net.  The strict inclusion is therefore failure of Haag duality, proving (ii), while
(iii) follows from Theorem~\ref{thm:haag-kraus}.  Part (i) follows from
Lemma~\ref{lem:wedge-zero-defect}.
\end{proof}

\subsection{Bosonic second quantization}
\label{subsec:weyl-dichotomy}

Let $\mathfrak h$ be a separable complex Hilbert space and $K\subset\mathfrak h$ a closed real
linear subspace.  Write $\mathcal R(K)$ for the von Neumann algebra on symmetric Fock space
generated by the Weyl unitaries $W(f)$, $f\in K$, with
\begin{equation}
W(f)W(g)=e^{-\frac{i}{2}\omega(f,g)}W(f+g),
\qquad
\omega(f,g)=\operatorname{Im}\langle f,g\rangle,
\label{eq:weyl-convention}
\end{equation}
and let $K':=\{g\in\mathfrak h:\omega(g,f)=0\ \text{for all }f\in K\}$ be the symplectic polar.
The map $K\mapsto\mathcal R(K)$ is an isomorphism of complemented nets: it carries generated
real subspaces to generated von Neumann algebras and the polar to the commutant,
\begin{equation}
\mathcal R(K)'=\mathcal R(K') .
\label{eq:second-quantization-commutant}
\end{equation}
This goes back to Araki \cite{Araki1963} and Eckmann--Osterwalder
\cite{EckmannOsterwalder1973}; we use the formulation of Figliolini and Guido
\cite[Prop.~2.1]{FiglioliniGuido2000}.  From \eqref{eq:weyl-convention},
\begin{equation}
W(f)W(g)=e^{-i\omega(f,g)}W(g)W(f).
\label{eq:weyl-exchange}
\end{equation}

\begin{proposition}[Extremal Weyl witnesses]
\label{prop:Weyl-extremal-witnesses}
Let $M,N\subset B(\cH)$ contain one-parameter unitary groups $U(r)\in M$ and $V(s)\in N$
satisfying
\begin{equation}
U(r)V(s)=e^{-i\kappa rs}V(s)U(r),\qquad \kappa\neq0.
\label{eq:CCR-phase}
\end{equation}
Then:
\begin{enumerate}[label=(\roman*)]
\item choosing $r_0,s_0$ with $\kappa r_0s_0=\pi$ and writing $U=U(r_0)$, $V=V(s_0)$, the
self-adjoint contractions
\begin{equation}
A_0:=\frac{U-U^*}{2i},\qquad B_0:=\frac{V-V^*}{2i}
\label{eq:trigonometric-Weyl-witnesses}
\end{equation}
satisfy
\begin{equation}
\diam\sigma\bigl(-i[A_0,B_0]\bigr)=4;
\label{eq:trigonometric-Weyl-diameter}
\end{equation}
in particular $\Gamma(M,N)=4$.
\item $M$ and $N$ contain self-adjoint unitaries $A,B$ satisfying $AB=-BA$.  In particular,
\[
\Delta_{\rm sa}(M,N)=2.
\]
\end{enumerate}
\end{proposition}

\begin{proof}
The phase condition gives $UV=-VU$, hence $A_0B_0=-B_0A_0$.  Put $X=U^2$ and $Y=V^2$.
They commute.  If $S=\sigma(X,Y)\subset\mathbb T^2$ is their joint spectrum, then conjugation by
$V(t)$ rotates $X$ through every phase and fixes $Y$, while conjugation by $U(q)$ fixes $X$ and
rotates $Y$ through every phase.  Thus the nonempty compact set $S$ is invariant under the
transitive $\mathbb T^2$-action and therefore $S=\mathbb T^2$.  Since
\[
A_0^2=\frac{2-X-X^*}{4},\qquad B_0^2=\frac{2-Y-Y^*}{4},
\]
the joint spectral point $(-1,-1)$ gives $\|A_0^2B_0^2\|=1$, hence $\|A_0B_0\|=1$.  For
$h=-i[A_0,B_0]=-2iA_0B_0$ one has $\|h\|=2$, and conjugation by $V$ sends $h$ to $-h$.
The spectrum $\sigma(h)$ is compact, so its spectral radius $2$ is attained; symmetry then gives
$\{-2,2\}\subset\sigma(h)$, proving (i).

For (ii), with the same $U,V$, let $\chi:\mathbb T\to\{-1,1\}$ be the Borel function satisfying
$\chi(-z)=-\chi(z)$.  Borel functional calculus gives self-adjoint unitaries
$A=\chi(U)\in M$ and $B=\chi(V)\in N$.  Since $VUV^*=-U$, functional calculus gives
$VAV^*=-A$, equivalently $AVA=-V$.  Applying functional calculus once more,
$ABA=\chi(AVA)=\chi(-V)=-B$, so $AB=-BA$.  Proposition~\ref{prop:Gamma-Delta}, together
with (i), yields $\Delta_{\rm sa}(M,N)=2$.
\end{proof}

The first part reaches the maximal infinitesimal coefficient by trigonometric continuous
functional calculus; the second supplies Clifford symmetries for the exact finite-parameter
profile below.

\begin{proposition}[Exact finite-scale order profile for a Clifford pair]
\label{prop:exact-Clifford-profile}
Let $A,B\in B(\cH)$ be self-adjoint unitaries with $AB=-BA$, and put
$\alpha_s=\Ad_{e^{isA}}$, $\beta_t=\Ad_{e^{itB}}$ and $q=q(s,t)=|\sin s\,\sin t|$.  Then for all
$s,t\in\mathbb R$,
\begin{equation}
\|\alpha_s\circ\beta_t-\beta_t\circ\alpha_s\|_{\cb}=4q\sqrt{1-q^2}.
\label{eq:exact-Clifford-order-defect}
\end{equation}
In particular the infinitesimal coefficient at the identity is $4$.
\end{proposition}

\begin{proof}
Let $u=e^{isA}=\cos s+i\sin s\,A$ and $v=e^{itB}=\cos t+i\sin t\,B$.  Because $B$ maps the
$+1$-eigenspace of $A$ onto the $-1$-eigenspace and is unitary, the two eigenspaces are
unitarily isomorphic; identifying them through $B$ we may write $\cH\cong\mathbb C^2\otimes\cK$
with
\[
A=\begin{pmatrix}1&0\\0&-1\end{pmatrix}\otimes1,
\qquad
B=\begin{pmatrix}0&1\\1&0\end{pmatrix}\otimes1 .
\]
In this representation the group commutator $w=w_{A,B}(s,t)=u^*v^*uv$ of
\eqref{eq:group-commutator} is a $2\times2$ unitary tensored with $1_{\cK}$, with $\det w=1$ and
\[
\tfrac12\operatorname{Tr}_2(w)=1-2\sin^2s\,\sin^2t=1-2q^2 ,
\]
so $\sigma(w)=\{e^{i\theta},e^{-i\theta}\}$ with
\begin{equation}
\cos\theta=1-2q^2 .
\label{eq:Clifford-relative-angle}
\end{equation}
The smallest closed disc containing two points of $\mathbb T$ has radius half their distance, so
$\inf_{\lambda\in\mathbb C}\|w-\lambda1\|=|\sin\theta|$, and
Proposition~\ref{prop:finite-scale-spectral} gives
\[
C_{A,B}(s,t)=2|\sin\theta|=2\sqrt{1-(1-2q^2)^2}=4q\sqrt{1-q^2}.
\]
This covers the degenerate cases $q\in\{0,1\}$, where $w=\pm1$ and both sides vanish.  The
infinitesimal statement follows from $q(s,t)/|st|\to1$ as $(s,t)\to(0,0)$.
\end{proof}

\begin{remark}[Finite-scale range]
\label{rem:finite-scale-range}
Along the slice $t=\pi/2$ formula \eqref{eq:exact-Clifford-order-defect} reduces to
\[
C_{A,B}(s,\pi/2)=2|\sin 2s| ,
\]
so a single Clifford pair realizes continuously every cb distance in $[0,2]$, which is the
largest possible distance between two unitary channels.  The infinitesimal coefficient therefore
does not determine the finite-parameter profile.
\end{remark}

\begin{proposition}[Bosonic second-quantization alternative]
\label{prop:second-quantization-dichotomy}
Let $K_1,K_2\subset\mathfrak h$ be closed real linear subspaces and $M_j=\mathcal R(K_j)$.
\begin{enumerate}[label=(\roman*)]
\item If $\omega(K_1,K_2)=0$ then $[M_1,M_2]=0$, and $C_{a,b}\equiv0$ for all bounded self-adjoint
$a\in M_1$, $b\in M_2$.
\item If $\omega(K_1,K_2)\neq0$ then $M_1$ and $M_2$ contain self-adjoint unitaries $A,B$ with
$AB=-BA$, whose channel flows satisfy
\begin{equation}
\bigl\|\Ad_{e^{isA}}\Ad_{e^{itB}}-\Ad_{e^{itB}}\Ad_{e^{isA}}\bigr\|_{\cb}=4q\sqrt{1-q^2},
\qquad q=|\sin s\,\sin t| .
\label{eq:second-quantization-finite-profile}
\end{equation}
Consequently $\Gamma(M_1,M_2)=4$ and $\Delta_{\rm sa}(M_1,M_2)=2$.
\end{enumerate}
\end{proposition}

\begin{proof}
If $\omega(K_1,K_2)=0$ then $K_2\subset K_1'$, so $M_2\subset\mathcal R(K_1')=M_1'$ by
\eqref{eq:second-quantization-commutant}, giving (i).  If $\omega(f,g)=\kappa\neq0$ for some
$f\in K_1$, $g\in K_2$, then $U(r)=W(rf)\in M_1$ and $V(s)=W(sg)\in M_2$ satisfy
\eqref{eq:CCR-phase} by \eqref{eq:weyl-exchange}.  Proposition~\ref{prop:Weyl-extremal-witnesses}
gives both $\Gamma(M_1,M_2)=4$ from
trigonometric witnesses and the required Clifford symmetries.  Proposition~\ref{prop:exact-Clifford-profile} gives
\eqref{eq:second-quantization-finite-profile}, and Proposition~\ref{prop:Gamma-Delta} yields
$\Delta_{\rm sa}(M_1,M_2)=2$.
\end{proof}

\begin{corollary}[Weyl-net criterion]
\label{cor:weyl-net-dichotomy}
Suppose the Haag--Kastler net is a bosonic second-quantization net, $\cA(O)=\mathcal R(K(O))$
with $K(O)$ closed real one-particle localization subspaces.  For an open region $R$ put
$K(R):=\overline{\operatorname{span}_{\mathbb R}}\{K(V):V\subset R\ \text{a double cone}\}$ and
$K^d(O):=K(O')'$.  Then $\cA^d(O)=\mathcal R(K^d(O))$, and for every spacelike pair exactly one
of the following holds:
\begin{enumerate}[label=(\alph*)]
\item $\omega(K^d(O_1),K^d(O_2))=0$, and all channel-order profiles between the two dual algebras
vanish;
\item $\omega(K^d(O_1),K^d(O_2))\neq0$, and the two dual algebras admit trigonometric
extremal witnesses as well as a Clifford pair with the profile
\eqref{eq:second-quantization-finite-profile}.
\end{enumerate}
In particular $\Gamma_{\cA}(O_1,O_2)\in\{0,4\}$, and failure of essential duality forces the
second alternative at some spacelike pair.
\end{corollary}

\begin{proof}
Preservation of generated subspaces and joins gives $\cA(O')=\mathcal R(K(O'))$, and taking
commutants with \eqref{eq:second-quantization-commutant} gives $\cA^d(O)=\mathcal R(K^d(O))$.
The alternatives are Proposition~\ref{prop:second-quantization-dichotomy}.  If essential duality
fails, Theorem~\ref{thm:quantitative-essential}(iii) supplies a spacelike pair with positive
defect, at which alternative (a) is excluded.
\end{proof}

\begin{remark}[Local reversible cb dichotomy for second quantization]
\label{rem:weyl-cb-tangent}
At a spacelike pair in Corollary~\ref{cor:weyl-net-dichotomy}, alternative (a) gives
$\Omega^{\mathrm{rev}}_\varepsilon=0$ for $0<\varepsilon<2$.  In alternative (b),
$\Gamma_{\cA}(O_1,O_2)=4$ and therefore
\[
\lim_{\varepsilon\downarrow0}
\frac{\Omega^{\mathrm{rev}}_\varepsilon(\cA^d(O_1),\cA^d(O_2))}{\varepsilon^2}=1.
\]
For the Clifford witnesses this limit is explicit: with
$r_\varepsilon=\arcsin(\varepsilon/2)$,
\[
C_{A,B}(r_\varepsilon,r_\varepsilon)
=\varepsilon^2\sqrt{1-\frac{\varepsilon^4}{16}}.
\]
\end{remark}

\begin{corollary}[Extremal spacelike pair without essential duality]
\label{cor:Yngvason-extremal}
Let $\cA$ be one of the generalized free field nets of \cite{Yngvason1994} for which essential
duality fails.  Then there are spacelike-separated double cones $O_1,O_2$ and self-adjoint
unitaries $A\in\cA^d(O_1)$, $B\in\cA^d(O_2)$ with $AB=-BA$, so that
\begin{equation}
\bigl\|\Ad_{e^{i\tau A}}\Ad_{e^{i\sigma B}}-\Ad_{e^{i\sigma B}}\Ad_{e^{i\tau A}}\bigr\|_{\cb}
=4q\sqrt{1-q^2},
\qquad q=|\sin\tau\,\sin\sigma| .
\label{eq:Yngvason-exact-order-defect}
\end{equation}
This pair has infinitesimal defect $4$ and admits finite parameters at which the two orders of
composition are at the maximal cb distance $2$.
\end{corollary}

\begin{proof}
The models of \cite[Secs.~2 and 7]{Yngvason1994} are generalized free fields presented in
Fock--Weyl form, hence bosonic second-quantization nets in the sense of
Corollary~\ref{cor:weyl-net-dichotomy}; see also \cite[Sec.~2]{GaierYngvason2000}.
Since essential duality fails, Corollary~\ref{cor:weyl-net-dichotomy} places some spacelike pair
in alternative (b), which
supplies $A$ and $B$; the profile is
Proposition~\ref{prop:exact-Clifford-profile}, and Remark~\ref{rem:finite-scale-range} gives the
finite-parameter statement.
\end{proof}

\subsection{Energy regularity in translation-covariant second quantization}
\label{subsec:energy-regular-weyl}

Assume now that the one-particle localization net is translation covariant,
$u(a)K(O)=K(O+a)$, with
\[
u(te_0)=e^{itH_1},\qquad H_1\ge0,
\]
and take $G=\mathrm d\Gamma(H_1)$ on symmetric Fock space.  The scalar phase in the Weyl
exchange relation disappears at the level of the displacement automorphisms themselves; the
nonzero order defect in Proposition~\ref{prop:Weyl-extremal-witnesses} is carried by the bounded
trigonometric generators built from those Weyl operators.  The next result controls these
generators with respect to $G$.

\begin{lemma}[Weyl graph-energy bound]
\label{lem:weyl-graph-energy}
If $f\in\operatorname{dom}H_1^{1/2}$, then, for every $t\in\mathbb R$ and
$\Psi\in\operatorname{dom}G^{1/2}$,
\begin{equation}
\bigl\|G^{1/2}W(tf)\Psi\bigr\|
\le
\bigl\|G^{1/2}\Psi\bigr\|
+\frac{|t|}{\sqrt2}\,\bigl\|H_1^{1/2}f\bigr\|\,\|\Psi\|.
\label{eq:weyl-graph-energy}
\end{equation}
Consequently, for
\[
A_{r,f}:=\frac{W(rf)-W(-rf)}{2i},\qquad
c_{r,f}:=\frac{|r|}{\sqrt2}\,\|H_1^{1/2}f\|,
\]
one has
\begin{equation}
\|G^{1/2}A_{r,f}\Psi\|
\le\|G^{1/2}\Psi\|+c_{r,f}\|\Psi\|.
\label{eq:trig-witness-graph-bound}
\end{equation}
The unitary group $e^{-isA_{r,f}}$ is $G$-energy-limited.
\end{lemma}

\begin{proof}
Choose an orthonormal basis $(e_j)\subset\operatorname{dom}H_1$.  On the algebraic Fock
core,
\[
\|G^{1/2}\Psi\|^2
=\sum_j\|a(H_1^{1/2}e_j)\Psi\|^2.
\]
The Weyl displacement relation shifts $a(H_1^{1/2}e_j)$ by a scalar of modulus
$|t|\,|\langle e_j,H_1^{1/2}f\rangle|/\sqrt2$.  Minkowski's inequality in the Hilbert direct
sum and Parseval's identity give \eqref{eq:weyl-graph-energy}.  Closure extends the estimate to
$\operatorname{dom}G^{1/2}$, and the estimates for $t$ and $-t$ show that $W(tf)$ preserves
that domain.

Averaging the estimates for $W(rf)$ and $W(-rf)$ gives
\eqref{eq:trig-witness-graph-bound}.  Hence $A_{r,f}$ is bounded for the
$G^{1/2}$-graph norm, so $e^{-isA_{r,f}}$ is a graph-norm strongly continuous group on
$\operatorname{dom}G^{1/2}$.  With
$x=\|G^{1/2}\Psi\|$ and $y=\|\Psi\|$,
\[
\left|i\left(
\langle G^{1/2}\Psi,G^{1/2}A_{r,f}\Psi\rangle
-\langle G^{1/2}A_{r,f}\Psi,G^{1/2}\Psi\rangle
\right)\right|
\le2x(x+c_{r,f}y)
\le3x^2+c_{r,f}^2y^2.
\]
Thus the quadratic-form criterion \cite[Thm.~3.7]{vanLuijk2025EnergyLimited} applies, for
example, with $\omega=3$ and $E_0=c_{r,f}^2/3$, and gives $G$-energy-limitedness.
\end{proof}

\begin{theorem}[Energy-regular extremal Weyl witnesses]
\label{thm:energy-regular-weyl}
Let $O\mapsto\cA(O)=\mathcal R(K(O))$ be a translation-covariant bosonic
second-quantization net with positive one-particle time-translation generator $H_1$.  If
essential duality fails, then there are spacelike-separated double cones
$\widetilde O_1,\widetilde O_2$
and vectors
\[
f\in K^d(\widetilde O_1)\cap\bigcap_{m\ge1}\operatorname{dom}H_1^m,
\qquad
g\in K^d(\widetilde O_2)\cap\bigcap_{m\ge1}\operatorname{dom}H_1^m
\]
with $\omega(f,g)\ne0$.  If $r,s\in\mathbb R$ are chosen so that
$\omega(f,g)rs=\pi$, then
\[
A_{r,f}=\frac{W(rf)-W(-rf)}{2i}\in\cA^d(\widetilde O_1),\qquad
B_{s,g}=\frac{W(sg)-W(-sg)}{2i}\in\cA^d(\widetilde O_2)
\]
are self-adjoint contractions satisfying
\begin{equation}
\diam\sigma\bigl(-i[A_{r,f},B_{s,g}]\bigr)=4,
\label{eq:energy-regular-extremal}
\end{equation}
and both unitary groups $e^{-itA_{r,f}}$ and $e^{-iuB_{s,g}}$ are
$G$-energy-limited.
\end{theorem}

\begin{proof}
By Corollary~\ref{cor:weyl-net-dichotomy}, failure of essential duality gives spacelike
double cones $O_1,O_2$ and vectors $f_0\in K^d(O_1)$, $g_0\in K^d(O_2)$ with
$\omega(f_0,g_0)\ne0$.  Choose a wedge $W$ with $O_1\subset W$ and $O_2\subset W'$
as in \eqref{eq:wedge-separation}.  There is a spacelike translation vector $v$ directed into
$W$ such that $\overline W+\varepsilon v\subset W$ and
$\overline{W'}-\varepsilon v\subset W'$ for $\varepsilon>0$; for the standard right wedge
one may take the positive first spatial coordinate vector.  Translation covariance of
$K^d$ gives
\[
f_\varepsilon:=u(\varepsilon v)f_0\in K^d(O_1+\varepsilon v),\qquad
g_\varepsilon:=u(-\varepsilon v)g_0\in K^d(O_2-\varepsilon v).
\]
Strong continuity implies
$\omega(f_\varepsilon,g_\varepsilon)\to\omega(f_0,g_0)$ as
$\varepsilon\downarrow0$, so fix $\varepsilon>0$ for which this pairing is nonzero.
The translated double cones now have closures strictly inside $W$ and $W'$.  Hence there are
spacelike double cones $\widetilde O_1\Subset W$, $\widetilde O_2\Subset W'$ and
$\delta>0$ such that
\[
O_1+\varepsilon v+te_0\subset\widetilde O_1,\qquad
O_2-\varepsilon v+te_0\subset\widetilde O_2
\qquad(|t|<\delta).
\]

Let $(\chi_n)$ be a real $C_c^\infty(\mathbb R)$ approximate identity with integral one and
support contained in $(-\delta,\delta)$, shrinking to zero.  Set
\[
f_n=\int\chi_n(t)u(te_0)f_\varepsilon\,dt,\qquad
g_n=\int\chi_n(t)u(te_0)g_\varepsilon\,dt.
\]
For every $t$ in the support, covariance and isotony put the integrands in
$K^d(\widetilde O_1)$ and $K^d(\widetilde O_2)$, respectively.  These are closed real
subspaces and $\chi_n$ is real, hence
$f_n\in K^d(\widetilde O_1)$ and $g_n\in K^d(\widetilde O_2)$.  Writing $\varphi_n(\lambda)=\int\chi_n(t)e^{it\lambda}\,dt$, spectral calculus gives
$f_n=\varphi_n(H_1)f_\varepsilon$ and similarly for $g_n$.  Since $\varphi_n$ is a Schwartz
function,
\[
f_n,g_n\in\bigcap_{m\ge1}\operatorname{dom}H_1^m.
\]
The approximate-identity property gives $f_n\to f_\varepsilon$ and
$g_n\to g_\varepsilon$ in norm, and therefore
$\omega(f_n,g_n)\to\omega(f_\varepsilon,g_\varepsilon)\ne0$.  Fix $n$ with nonzero
pairing and write $f=f_n$, $g=g_n$.

Proposition~\ref{prop:Weyl-extremal-witnesses} gives
\eqref{eq:energy-regular-extremal} after choosing $r,s$ with $\omega(f,g)rs=\pi$.
Lemma~\ref{lem:weyl-graph-energy} applies to both generators and proves the final statement.
\end{proof}

\section{Conclusion}

The causal-complement channel class is
\[
\NS_{\cA}(O)=\Chan(\cA^d(O)).
\]
Accordingly, Haag duality identifies this class with $\Chan(\cA(O))$, whereas essential
duality is order-independence of the corresponding classes at spacelike separation.  The
small-parameter order defect has coefficient $\diam\sigma(-i[a,b])$, and optimization gives
\[
\Gamma(M,N)=2\Delta_{\rm sa}(M,N),\qquad
\lim_{\varepsilon\downarrow0}
\frac{\Omega^{\rm rev}_\varepsilon(M,N)}{\varepsilon^2}
=\frac14\Gamma(M,N).
\]
For the full inner-channel classes, the coefficient $\Lambda(M,N)$ is the optimal normalized
order defect and has the same commutativity zero set.

With a positive reference observable $G$, the mean-input-energy-constrained coefficient obeys
\[
\Gamma_{G,E}(M,N)=0\iff[M,N]=0,
\qquad
\Gamma_{G,E}(M,N)\uparrow\Gamma(M,N).
\]
Spectral-edge accessibility provides a sufficient finite-energy recovery threshold, but the
general hypotheses impose no rate: Proposition~\ref{prop:no-uniform-energy-rate} gives
arbitrarily slow recovery even for compact-resolvent $G$ with finite spectral multiplicities,
ground energy zero, and first positive eigenvalue one.

For bosonic second quantization, the dual localization subspaces give a sharp dichotomy:
vanishing symplectic pairing yields zero spacelike defect, while nonzero pairing yields
$\Gamma=4$ together with the Clifford finite-parameter profile.  Wedge-dual generalized free
fields realize the zero-defect branch despite bounded-region Haag-duality gaps, and generalized
free fields failing essential duality realize the nonzero branch.  In translation-covariant
second-quantization nets with positive one-particle time generator, the extremal witnesses can
also be chosen from vectors lying in every power domain of that generator, with bounded witness
flows that are energy-limited relative to the second-quantized time-translation energy.
\begingroup
\small
\setlength{\bibsep}{2pt}
\bibliographystyle{unsrtnat}
\bibliography{refs}
\endgroup
\end{document}